\documentclass[11pt,letterpaper,reqno]{amsart}

\usepackage{mathtools}
\usepackage{amsthm,amssymb}
\usepackage[usenames,dvipsnames]{xcolor}
\usepackage{graphicx}
\usepackage{tikz}

\usepackage[shortlabels]{enumitem}
\usepackage{hyperref}
\hypersetup{colorlinks = true, linkcolor = blue, citecolor = blue, urlcolor = blue}

\allowdisplaybreaks

\usepackage[margin=1in]{geometry}
\usepackage[textsize=small]{todonotes}

\usepackage{bookmark}

\newtheorem{theorem}{Theorem}
\newtheorem*{theorem*}{Theorem}
\newtheorem{conjecture}[theorem]{Conjecture}
\newtheorem{corollary}[theorem]{Corollary}
\newtheorem{lemma}[theorem]{Lemma}

\newtheorem{proposition}[theorem]{Proposition}

\newtheorem{claim}[theorem]{Claim}
\newtheorem*{claim*}{Claim}

\newcommand{\C}{\mathcal{C}}
\newcommand*{\cP}{\mathcal{P}}
\newcommand{\BIS}{\#{}BIS}
\def\gam{\gamma}
\newcommand{\dist}{\mathrm{dist}}
\newcommand{\nm}{\operatorname{nm}}
\def\eps{\epsilon}

\begin{document}

\title{Algorithms for the ferromagnetic Potts model on expanders}

\author[C.\ Carlson]{Charlie Carlson}
\address{Computer Science Department \\ University of California, Santa Barbara}
\email{CharlieAnneCarlson@ucsb.edu}

\author[E.\ Davies]{Ewan Davies}
\address{Department of Computer Science \\ Colorado State University}

\email{research@ewandavies.org}

\author[N.\ Fraiman]{Nicolas Fraiman}
\address{Department of Statistics and Operations Research \\ University of North Carolina at Chapel Hill}
\email{fraiman@unc.edu}

\author[A.\ Kolla]{Alexandra Kolla}
\address{Computer Science and Engineering Department\\ University of California, Santa Cruz}
\email{akolla@ucsc.edu}

\author[A.\ Potukuchi]{Aditya Potukuchi}
\address{Department of Electrical Engineering \& Computer Science, York University}
\email{apotu@yorku.ca}

\author[C.\ Yap]{Corrine Yap}
\address{School of Mathematics \\ Georgia Institute of Technology}
\email{math@corrineyap.com}

\thanks{NF is supported in part by NSF grant CCF-1934964, AK is supported by NSF CAREER grant 1452923, AP is supported in part by NSF grant CCF-1934915, CY is supported in part by NSF grant CCF-1814409 and NSF grant DMS-1800521. An extended abstract of this work appeared at FOCS 2022.}

\begin{abstract}
We give algorithms for approximating the partition function of the ferromagnetic $q$-color Potts model on graphs of maximum degree $d$. Our primary contribution is a fully polynomial-time approximation scheme for $d$-regular graphs with an expansion condition at low temperatures (that is, bounded away from the order-disorder threshold). The expansion condition is much weaker than in previous works; for example, the expansion exhibited by the hypercube suffices. \\
The main improvements come from a significantly sharper analysis of standard polymer models; we use extremal graph theory and applications of Karger's algorithm to count cuts that may be of independent interest. It is \#BIS-hard to approximate the partition function at low temperatures on bounded-degree graphs, so our algorithm can be seen as evidence that hard instances of \#BIS are rare. \\
We also obtain efficient algorithms in the Gibbs uniqueness region for bounded-degree graphs. While our high temperature proof follows more standard polymer model analysis, our result holds in the largest known range of parameters $d$ and $q$.
\end{abstract}

\maketitle

\section{Introduction}

The $q$-state Potts model on a graph $G=(V,E)$ at inverse temperature $\beta$ is given by the partition function
\begin{equation}
\label{eqn:pottspartition}
Z_G(q, \beta) = \sum_{\sigma: V\to [q]}e^{\beta m(G,\sigma)},
\end{equation}
where the sum is over all assignments of the spins $[q]=\{1,2,\dotsc,q\}$ to the vertices of $G$, and $m(G,\sigma)$ is the number of edges such that both endpoints receive the same spin under $\sigma$. 
The Potts model has been of continued interest in combinatorics and computer science, most notably because of its direct relations to notions such as $q$-cuts and graph colorings. 
The Potts model also arises in physics and other areas as a generalization of the Ising model.
The main computational question associated to the Potts model is to approximate the partition function. The model is a canonical example of a Markov random field, and hence an excellent testbed for algorithmic techniques.

The Potts model as above is \emph{ferromagnetic} when $\beta > 0$ and \emph{antiferromagnetic} when $\beta < 0$.
We are interested in the ferromagnetic model because for $q\ge 3$, at low temperatures approximating the partition function is \BIS{}-hard~\cite{GJ12,GSVY16}. 
The complexity class known as \BIS{} consists of problems that are equivalent to counting the number of independent sets in bipartite graphs~\cite{DGGJ04}, and understanding the complexity of approximating problems in \BIS{} is a longstanding problem.
More precisely, approximating the partition function of the ferromagnetic Potts model for fixed $\beta$ is \BIS-hard~\cite{GJ12}, and restricted to bounded-degree graphs it is \BIS-hard at low temperatures~\cite{GSVY16}. 
That is, for any $d\ge 3$ and $q\ge 3$, for any $\beta > \beta_o(q,d)$ it is \BIS-hard to approximate $Z_G(q,\beta)$ on graphs of maximum degree $d$, where $\beta_o$ is given by 
\begin{equation}
    \beta_o(q,d) = \ln\left(\frac{q-2}{(q-1)^{1-2/d}-1}\right).
\end{equation}
The parameter $\beta_o$ has a precise definition as the \emph{order-disorder} threshold of the Potts model on the random $d$-regular graph, and we note that it is not the same as the Gibbs uniqueness phase transition on the $d$-regular tree. 
See e.g.~\cite{GSVY16,CGG+22} for a discussion of these thresholds.
The Gibbs measure $\mu_{G,q,\beta}$ for the Potts model on the spin assignments of $G$ is given by
\[ \mu_{G,q,\beta}(\sigma) = e^{\beta m(G,\sigma)} / Z_G(q,\beta) \]
and, alongside approximating $Z$, algorithms that approximately sample from $\mu_{G,q,\beta}$ are of interest in theoretical computer science. 
By self-reducibility, the standard problems of approximating the partition function and approximately sampling from the Gibbs measure are equivalent.

Several recent works have given algorithms for \BIS-hard problems on random graphs or graphs with strong expansion~\cite{JKP20,HJP20,GGS21b,CGG+21}. 
One consequence of such research is a restriction on possible hard instances of \BIS{}, since under somewhat general conditions these works give efficient algorithms. 
In the absence of strong evidence that \BIS{} admits a general polynomial-time approximation algorithm, a weaker but intriguing prospect is that of a subexponential-time algorithm. 
The recent algorithm of Jenssen, Perkins, and Potukuchi~\cite{JPP22} counts independent sets in $d$-regular bipartite graphs in subexponential time when $d$ grows with the number of vertices of the instance. 
Our work is motivated in part by the question of whether a subexponential-time approximation algorithm for the Potts partition function can be found (at temperatures known to be \BIS-hard).

In general, the interaction between probabilistic phase transitions or thresholds such as $\beta_o$ and the computational complexity of approximating partition functions has received substantial attention. 
For example, while the (Gibbs uniqueness) phase transition on the infinite $d$-regular tree corresponds to an NP-hardness threshold in graphs of maximum degree $d$ for antiferromagnetic 2-spin models such as the hard-core model and the Ising model~\cite{Sly10,SS14,GSV16}, for the ferromagnetic Potts model on the random $d$-regular graph neither the Gibbs uniqueness phase transition nor the order-disorder threshold correspond to a computational threshold~\cite{HJP20} when $q$ is large enough in terms of $d$. 
The methods of~\cite{HJP20} rely on large $q$, and our work addresses the question of how one can work with smaller $q$.

While our primary objective is to give algorithms for approximating $Z_G(q,\beta)$ in the \BIS-hardness region (i.e.\ low temperatures), we also investigate phase transitions subject to much weaker conditions than those satisfied by the random regular graph. To this end, we develop a high-temperature algorithm for general $d$-regular instances that works up to some $\beta_1(q,d)$ such that $\beta_1\sim \beta_o$ as $d\to \infty$. 
Since the Gibbs uniqueness threshold is less than $\beta_o$, our work gives a polynomial-time approximation algorithm in the Gibbs uniqueness region for regular graphs and does so for the largest known range of $q$.

\subsection{Our results}

Let $G = (V,E)$ be a $d$-regular graph on $n$ vertices. 
For a set of vertices $A\subset V$, let the \emph{edge-boundary} $\nabla(A)$ of $A$ be the set of edges of $G$ with precisely one endpoint in $A$. 
For $\eta > 0$, we say that a graph $G$ is an \emph{$\eta$-expander} if
\begin{equation}\label{eqn:expansion}
\text{every nonempty set}~A\subset V~\text{with}~|A|\le n/2~\text{satisfies}~|\nabla(A)| \geq \eta|A|.
\end{equation}
We note here that the $d$-dimensional hypercube $Q_d$ is a $1$-expander, and that $\eta$-expansion differs by a factor $d$ from the usual definition of edge expansion $\phi$ in an $n$-vertex $d$-regular graph, $\phi(G) = \min_{0 < |A| \le n/2}|\nabla(A)|/(d|A|)$.

For $z>0$ and $\delta\in [0,1)$, we say that $\hat z$ is a \emph{$\delta$-relative approximation} to $z$ if $1- \delta \le z/\hat z \le 1+\delta$. 
For a given $q \in \mathbb{N}$ and $\beta > 0$, a \emph{fully polynomial-time approximation scheme} (FPTAS) for $Z_G(q,\beta)$ is an algorithm that for every $\delta>0$ outputs a $\delta$-relative approximation to $Z_G(q,\beta)$ and runs in time polynomial in $n$ and $1/\delta$. 
A \emph{polynomial-time sampling scheme} for $\mu_{G,q,\beta}$ outputs a random spin assignment with distribution $\hat \mu$ that lies within $\delta$ total variation distance of $\mu_{G,q,\beta}$ and runs in time polynomial in $n$ and $1/\delta$. 
Our main theorem is the following.

\begin{theorem}\label{thm:mainpotts}
For every $\epsilon > 0$, there exists $d_0(\eps)$ such that for $d \ge d_0(\epsilon)$, $q \ge d^c$ where $c$ is an absolute constant, and positive $\beta \notin ((1-\epsilon)\beta_o, (1+\epsilon)\beta_o)$, there exist
\begin{enumerate}
\item an FPTAS for $Z_G(q,\beta)$, and
\item a polynomial-time sampling scheme for $\mu_{G,q,\beta}$
\end{enumerate}
for $G$ in the class of $d$-regular $2$-expander graphs, and for $G$ in the class of triangle-free $d$-regular $1$-expander graphs.
\end{theorem}

While the constants $2$ and $1$ marking the limit of our expansion range may seem arbitrary, they represent a subtle barrier in our proofs. All $n$-vertex $d$-regular graphs that are $2$-expanders or triangle-free $1$-expanders must have min-cut $d$ and the property that a set $A$ with $2\le |A|\le n/2$ vertices on one side has edge-boundary at least $2d-2$ (see~\eqref{eqn:mincut2}). We choose not to pursue any additional case analysis and parameter tradeoffs that might overcome this barrier. One option, as in~\cite{HJP20}, is to require stronger (e.g. $\Omega(d)$) expansion for small sets but we avoided this assumption for Theorem~\ref{thm:mainpotts}. It is also worth noting that we require $d_0(\epsilon) \geq e^{\Omega(1/\epsilon)}$.

In the low-temperature regime, $\beta > (1+\eps)\beta_o$, this is a significant improvement over a previous result of Jenssen, Keevash, and Perkins~\cite{JKP20} which required that $q \geq d^{\Omega(d)}$, as well as the graph being an $\Omega(d)$-expander when considering the same range of temperatures. Briefly, the statement in~\cite{JKP20} permits weaker $\eta$-expansion than $\eta=2$, but at the cost of larger $\beta$. 
We focus\footnote{Our work could also be used to prove a version of the low-temperature statement in Theorem~\ref{thm:mainpotts} for a strictly larger range of parameters than those handled in~\cite{JKP20} since we extend their method with a sharper analysis.} on obtaining an algorithm that that still applies when $\beta$ is close to $\beta_o$. 
For a $\delta$-relative approximation, the high-temperature algorithm presented in Theorem~\ref{thm:mainpotts} runs in time $\left(n/\delta\right)^{O_{\epsilon}(\ln d)}$ and the low-temperature algorithm runs in time $\left(n/\delta\right)^{O_{\epsilon}\left(d/\eta\right)}$.

For the high-temperature regime where $\beta\leq(1-\eps)\beta_o$, an expansion assumption is not required and the $d$-regular condition can be relaxed to maximum degree $d$. Moreover, there is a function $q_0(\eps)$ such that our method only requires $q\ge q_0(\eps)$. In particular, this gives an FPTAS to sample in the uniqueness regime (and a bit beyond) for large enough $q$ and $d$. This improves upon~\cite[Theorem~2.4]{BCH+20} which gives an FPTAS in the range $\beta \le \frac{3}{2}\ln(q)/d$ for all $d\ge 2$ when $q\ge \exp(\Omega(d\ln d))$ (see also a slightly different version in~\cite{CDK+20} that applies up to $(1-\eps)\beta_o$ but requires $q\ge \exp(\Omega(d^{3/2}\ln d))$).

The ``gap'' in allowed values of $\beta$ in Theorem~\ref{thm:mainpotts} is partly due to our focus on reducing the necessary lower bound on $q$. 
Such a gap should not be necessary (see~\cite{HJP20}) for large enough $q$, and it would be interesting to close the gap without strengthening the lower bound on $q$ in the statement.

Guided by results for other models in statistical physics, one might expect efficient approximation algorithms for $Z_G(q,\beta)$ to exist on graphs of maximum degree $d$ when $(q,\beta)$ lie in the Gibbs uniqueness region of the infinite $d$-regular tree, for any $d,q\ge 3$ (For $d=2$ the model is exactly solvable and $q=2$ gives the Ising model, which is rather well-understood by comparison).
Our methods give an FPTAS in this region for the largest-known range of $d$ and $q$.

\begin{theorem}\label{thm:FPTASuniqueness}
There exist absolute constants $d'$ and $q'$ such that for all $d\ge d'$ and $q\ge q'$, there is an FPTAS for $Z_G(q,\beta)$ on graphs of maximum degree $d$ when $\beta>0$ lies in the Gibbs uniqueness region of the $q$-color Potts model on the infinite $d$-regular tree. 
\end{theorem}

While our high-temperature result does not require the same novel combinatorial techniques (described below) as the low-temperature improvement, our work as a whole emphasizes some interesting aspects of the polymer models that we use to prove the above results. We discuss in a concluding section the ``critical window'' where $\beta$ is close to $\beta_o$ and the difficulties associated with lowering $q$.

One of our key innovations for the low-temperature result, which is of independent interest, is an upper bound on the number of sets in a $d$-regular $\eta$-expander with a prescribed edge-boundary size $b$. 
This is reminiscent of container theorems from extremal graph theory that bound the number of independent sets in some graph or hypergraph. 
The origin of container theorems can be traced back to Kleitman and Winston~\cite{KW80,KW82} but the ideas were developed significantly by Sapozhenko~\cite{Sap87,Sap01}, who specifically studied independent sets in expander graphs and the hypercube. 
Algorithmic applications of containers for counting independent sets in bipartite graphs were recently given by Jenssen and Perkins~\cite{JP20}, also with Potukuchi~\cite{JPP22}.
The type of bound we need does not seem to relate closely to an existing container theorem, and we prove the following novel result using ideas from randomized algorithms. 
We say that a subset $A$ of vertices in a graph $G$ is connected if the induced subgraph $G[A]$ is connected.

\begin{theorem}\label{thm:enum}
Let $G=(V,E)$ be a $d$-regular $\eta$-expander and let $u$ be a vertex in $G$. Then the number of connected sets $A\subset V$ of size at most $|V|/2$ such that $u \in A$ and $|\nabla(A)|=b$ is at most $d^{O\left((1+1/\eta)b/d\right)}$.
\end{theorem}

One way of interpreting the theorem is that we show that a large-girth graph (or, if you prefer, an infinite regular tree) is extremal for the number of connected sets with a given boundary $b$, we explain this heuristic in the overview below.

Our methods also give an idea of the typical structure of the Potts model on $\eta$-expanders below and above the order-disorder threshold $\beta_o$. More precisely, we show the following.

\begin{theorem}\label{thm:pt}
Let $\epsilon$, $d$, $q$, and $G$ be as in Theorem~\ref{thm:mainpotts}, with $n=|V(G)|$. As $n\to\infty$, for a coloring sampled from the ferromagnetic Potts model,
\begin{enumerate}
\item\label{part:pt1} for $\beta < (1 - \epsilon)\beta_o$, each color class has size $(1 + o(1))n/q$ with high probability, and
\item\label{part:pt2} for $\beta > (1+\epsilon)\beta_o$, there is a color class that contains at least $(1-o_d(1))n$ vertices with high probability.
\end{enumerate}
\end{theorem}

The main result of~\cite{HJP20} gives a much more complete description of Potts model (in fact, the more general random cluster model that we describe below) on $d$-regular locally tree-like graphs subject to an $\eta$-expansion condition of the form $\eta=\Omega(d)$ and a stronger small-set expansion condition. 
It would be interesting to increase the precision of our approach to give a similar description subject to weaker conditions, e.g.\ for the Potts model on the hypercube. 
We give an example of what our methods currently reveal about the Potts model on the hypercube below.

An equivalent definition of $Z$ is given by the Fortuin--Kasteleyn representation~\cite{FK72}
\begin{equation}\label{eq:FKrep} Z_G(q,\beta) = \sum_{A \subseteq E} q^{c(A)} (e^\beta - 1)^{|A|}, \end{equation}
where $c(A)$ is the number of connected components of the graph $(V,A)$ on the same vertex set as $G$ but with edge set $A$. 
In this form, $Z$ is known as the partition function of the random cluster model, and it is easy to see that $Z$ is a reparametrization of the Tutte polynomial of $G$. The random cluster model is a distribution on subsets of edges where the probability mass of $A \subseteq E$ is proportional to $q^{c(A)} (e^\beta - 1)^{|A|}$. It is standard to change parameters via $p=e^\beta-1$. 
As an illustration of the power of our new techniques, we give the following structural result for the components induced by an edge set drawn from the random cluster model on the discrete $d$-dimensional hypercube $Q_d$.

\begin{theorem}\label{thm:maincube}
Let $\eps>0$ and $d\ge d_0(\eps)$ for the $d_0$ in Theorem~\ref{thm:mainpotts}.
For a constant $c$ at least the minimum $c$ of Theorem~\ref{thm:mainpotts}, let $q = d^{c}$ be a positive integer. There is a sequence $p_o = p_o(d,q)$ such that in the random cluster model on the $d$-regular discrete hypercube $Q_d$, for every $\epsilon >0$,
\begin{enumerate}
\item\label{part:subcritical} when $p \leq  (1 - \epsilon)p_o$, every component has size at most $d$ with high probability, and
\item\label{part:supercritical} when $p \geq (1 + \epsilon)p_o$, there is a unique connected component of size $2^d\left(1 - o(1)\right)$ with high probability.
\end{enumerate}

\end{theorem}

Some remarks are in order here. First, the condition $q = d^c$ means that the number of colors grows (poly-logarithmically) with the number of vertices in the graph. This is admittedly a slightly different regime than the more common one where $q$ and $d$ are both fixed. This difference turns out to be inconsequential, and the behavior predicted by Theorem~\ref{thm:pt} ultimately holds. 

Second, we note that for any constant $c>0$ and $q=d^c$, $\ln (1 + p_o) \sim p_o \sim \frac{2\ln q}{d}$ as $d\to\infty$. 
When considering the random cluster model one should not need to restrict $q$ to the positive integers, though our methods currently require the condition that $q$ is an integer.  
It would be interesting to remove this restriction, and we discuss the barriers in the concluding section.

\subsection{An overview of techniques}

Our algorithms follow from the cluster expansion for abstract polymer models, which was first used to design sampling algorithms for the Potts model by Helmuth, Perkins and Regts~\cite{HPR19b}. The method has been studied intensively since, including further works concerning the ferromagnetic Potts model~\cite{JKP20,BCH+20,HJP20} and many other works tackling various Markov random fields.
At a high level, the method proceeds by transforming the partition function into a representation that corresponds to an abstract but particularly well-studied \emph{polymer model} from statistical mechanics. 
The strength of this approach lies in the fact that general conditions on these models give a wealth of probabilistic and algorithmic information.
The specific polymer models we rely on have appeared in several earlier works~\cite{JKP20,BCH+20,HJP20,CDK+20,CDK20}, and our key innovations are improved analyses of the models to obtain algorithms under significantly weaker assumptions. 
Let us split the rest of the discussion into the \emph{low-temperature} regime ($\beta \geq (1 +\epsilon)\beta_o$) and the \emph{high-temperature} regime ($\beta \leq (1 -\epsilon)\beta_o$).

\textbf{Low temperature.} This part contains our main contributions. Some of the difficulties we face boil down to understanding certain combinatorial questions that are of independent interest. Informally, two problems that arise are as follows. Let $G$ be a $d$-regular $\eta$-expander with $n$ vertices.
\begin{enumerate}[label = {(\alph*)}]
\item\label{enumproblem-boundary} What is the number of connected, induced subgraphs of $G$ containing a given vertex and with an edge-boundary of size $b$?
\item\label{enumproblem-colorings} What is the number of $q$-colorings of $G$ that have exactly $k$ non-monochromatic edges?
\end{enumerate}

Let us first quantitatively motivate the types of answers that can be expected for these questions. 
For~\ref{enumproblem-boundary}, suppose that $G$ is weakly ``locally tree-like'' in the sense that every connected set of $a \ll n$ vertices spans $O(a)$ edges and therefore has an edge-boundary of size $\Omega(d a)$. Working backwards, a naive heuristic is that a typical connected subset of $\Theta(b/d)$ vertices has an edge-boundary of size $b$. 
The locally tree-like property also means that the number of connected, induced subgraphs of $G$ containing a given vertex is roughly the number of trees with degree at most $d$ rooted at the vertex, and hence for subgraphs on $\Theta(b/d)$ vertices this number is roughly $d^{O(b/d)}$, we can answer~\ref{enumproblem-boundary} for locally tree-like graphs. The same argument also works for graphs satisfying condition~\eqref{eqn:expansion} (by noting that a set of boundary size $b$ may have at most $b/\eta$ vertices) to give a bound of $d^{O(b/\eta)}$.
Theorem~\ref{thm:enum}
gives a stronger upper bound for graphs satisfying our expansion condition~\eqref{eqn:expansion} by finding a small ``certificate'' for each connected, induced subgraph with a prescribed edge-boundary size. 
Since the certificate is small, there cannot be too many valid certificates and an upper bound on the number of sets we are interested in follows.
One of the key tools for the proof is a standard but elegant adaptation of Karger's randomized algorithm for finding minimum cuts~\cite{Kar93} to the problem of counting cuts; see Theorem~\ref{thm:karger}.
To the best of our knowledge, our application of this celebrated result in combinatorial optimization to analyze a counting algorithm is novel, and the technique may be of independent interest.
Some additional ideas used in the proof can be found in~\cite{PY22}.

For~\ref{enumproblem-colorings}, the heuristic is as follows: the simplest way that one can obtain a $q$-coloring of $G$ with $k$ non-monochromatic edges is by coloring all but $k/d$ randomly chosen vertices, which for small $k$ are likely to form an independent set, in the same color. The independent set on $k/d$ vertices can be colored arbitrarily with the other colors, giving the required $k$ non-monochromatic edges. This gives a lower bound of roughly $\binom{n}{k/d} q^{k/d} = (ndq/k)^{\Omega(k/d)}$. We provide another enumeration result (Lemma~\ref{lem:groundstate}) that again justifies this heuristic for graphs satisfying the expansion condition~\eqref{eqn:expansion}. 
Note that the emergence of problem~\ref{enumproblem-colorings} involving colorings means that it is not straightforward to generalize our low-temperature algorithm to the random cluster model (with non-integer $q$).

In a calculation establishing that the polymer models we use provide efficient algorithms, we require upper bounds on $Z_\gam(q-1,\beta)$ where $\gam$ is a connected, induced subgraph of $G$. 
Previous works used crude bounds here, and we use better bounds from~\cite{SSSZ20} that are tight when $\gamma$ is isomorphic to $K_{d+1}$, or in the triangle-free case, complete bipartite graphs $K_{d,d}$. 
This is still not ideal as an $\eta$-expander cannot contain such subgraphs, but it is generally difficult to prove improved bounds, or even identify graph properties that would allow an improved bound in such problems~\cite{PY22}.
This is another place in which our methods require the Potts model, as such bounds are not known for the random cluster model. 
In this way, our algorithmic work motivates new questions in extremal graph theory; see Conjectures~\ref{conj:cliqueextremal} and~\ref{conj:bicliqueextremal}.

\textbf{High temperature.} 
The main idea to handle this regime is to switch to the random cluster model. The primary reason for this is to access a convenient representation as a polymer model. 
Using the Edwards--Sokal coupling~\cite{ES88}, one can view both of these models on the same probability space. 
The coupling gives an algorithm for converting a random cluster sampling algorithm to a Potts model sampling algorithm and vice versa. 
The polymer model we use was also generalized to a partition function related to Unique Games in~\cite{CDK+20}. 
The analysis given in this paper improves upon the analyses in~\cite{BCH+20} and~\cite{CDK+20}. 
Moreover, this part of the proof works for \emph{any} graph of maximum degree $d$ (i.e.\ we can dispense with $d$-regularity and with an expansion assumption), a fact present in these earlier works which does not seem to be well-known in general.

\subsection{Related work}
Algorithmic applications of the cluster expansion originate in~\cite{HPR19b}, and for the ferromagnetic Potts model in particular there are two works of note that establish algorithms at all temperatures on subgraphs of the integer lattice $\mathbb{Z}^d$~\cite{BCH+20} and on regular graphs with strong expansion~\cite{HJP20}. 
These results rely on $q$ being at least exponentially large in $d$ (and even larger if the expansion guarantee is as weak as that of Theorem~\ref{thm:mainpotts}), which is a substantial drawback but in return the approaches work for all temperatures. 
The focus of our work is slightly different, and essentially we pay for working with smaller $q$ that grows only polynomially with $d$ (along with much weaker expansion) by accepting a small gap in our techniques around the threshold $\beta_o$.

Our main tool is the cluster expansion for abstract polymer models, as in~\cite{HPR19b,JKP20,BCH+20,HJP20} (and many more works), where deterministic approximation algorithms for partition functions follow from convergent cluster expansions by approximating a truncated series. 
See also Barvinok's work (e.g.\ the monograph~\cite{Bar16}) which follows a similar series approximation approach.
A twist on the cluster expansion method originating in~\cite{CGG+21} uses a Markov chain to sample from an abstract polymer model.
The approaches of~\cite{JKP20,CGG+21} were combined and generalized in~\cite{GGS21b} to give algorithms for general spin systems on bipartite expanders. 
We do not pursue Markov-chain based sampling from polymer models here as the conditions required in~\cite{CGG+21,GGS21b,BCP22} (i.e.\ upper bounds on polymer weights) are stronger than what we can guarantee with our current methods, and we do not want to restrict our attention to bipartite graphs.

A popular approach to approximate counting and sampling is to simulate a Markov chain whose stationary distribution is the Gibbs measure of interest. 
For the Potts model there are two main Markov chains, the \emph{Glauber dynamics} and the \emph{Swendsen--Wang dynamics}. 
The main result of~\cite{BGP16} shows for graphs of maximum degree $d$ that Glauber dynamics is rapidly mixing for $\beta$ up to a certain threshold and slow mixing for $\beta$ at least some slightly larger threshold.
For Glauber dynamics on the random regular graph there is an improved understanding: the Gibbs uniqueness threshold marks the change from rapid mixing to torpid mixing~\cite{BG21a,CGG+22}.
For comparison with our results, we note that each of the three thresholds discussed above lie below $\beta_o$ and are asymptotic to $\ln(q)/d$ as $d\to\infty$ (which is asymptotically half of $\beta_o\sim 2\ln(q)/d$ as $d\to\infty$). It should be noted though, that these results hold for {\em every} $q$, while the results in this paper are restricted to large $q$.
The picture for Swendsen--Wang dynamics on the random regular graph is less well-understood; see~\cite{CGG+22} and the references therein for details.  
An alternative approach to sampling from the ferromagnetic Potts model (and random cluster model) on the random regular graph in the Gibbs uniqueness regime of the ferromagnetic Potts model was given in~\cite{BGG+20}.

\subsection{Organization}
Section~\ref{sec:prelim} is an introduction to abstract polymer models and their use in obtaining approximate counting and sampling algorithms. 
In Section~\ref{sec:lowtemp} we prove our algorithmic result---Theorem~\ref{thm:mainpotts}---for the low-temperature case ($\beta > (1 + \epsilon)\beta_o$) using the solutions to problems~\ref{enumproblem-boundary},~\ref{enumproblem-colorings}, and extremal results on partition functions (Theorem~\ref{thm:enum}, Lemma~\ref{lem:maingroundstate}, and Lemmas~\ref{lem:extremal} and~\ref{lem:TFextremal} respectively). Section~\ref{sec:enum} is dedicated to the proof of our container-like result Theorem~\ref{thm:enum}, solving problem~\ref{enumproblem-boundary}; and Section~\ref{sec:maingroundstate} to the proof of Lemma~\ref{lem:maingroundstate}, solving problem~\ref{enumproblem-colorings}. In Section~\ref{sec:extremal}, we prove the extremal results, Lemmas~\ref{lem:extremal} and~\ref{lem:TFextremal}. Section~\ref{sec:hightemp} consists of the proof of the high-temperature case ($\beta < (1-\epsilon)\beta_o$) in Theorem~\ref{thm:mainpotts}. In Section~\ref{sec:pt} we use our results to characterize the structure of the Potts model, proving Theorem~\ref{thm:pt}. 
Finally, Section~\ref{sec:rccube} is dedicated to the structure of the random cluster model on the hypercube, namely the proof of Theorem~\ref{thm:maincube}. We conclude with a discussion of further directions to pursue.

\section{Preliminaries}\label{sec:prelim}

\subsection{Abstract polymer models}

An \emph{abstract polymer model}~\cite{GK71} is given by a set $\cP$ of polymers, a compatibility relation $\sim$ on polymers, and a weight $w_\gam$ associated to each polymer $\gam\in \cP$. 
In applications, polymers might be combinatorial objects such as vertex subsets or connected subgraphs of a graph.
Let $\Omega$ denote the collection of all sets of pairwise compatible polymers, including the empty set (i.e.\ the set of independent sets in the graph on $\cP$ with edges between incompatible polymers). 
Then the polymer model partition function is 
\[ \Xi = \sum_{\Lambda\in\Omega}\prod_{\gam\in\Lambda}w_\gam. \]
The associated Gibbs measure $\nu$ on $\Omega$ is given by $\nu(\Lambda) = \prod_{\gam\in\Lambda}w_\gam / \Xi$, in much the same way as the Gibbs measures associated to the partition function of the Potts and random cluster models.
It may be helpful to observe that $\Xi$ is the multivariate independence polynomial of the graph on $\cP$ with edges between incompatible polymers, evaluated at the polymer weights given by $w$, and hence $\nu$ is the associated hard-core measure.

The \emph{cluster expansion} of $\Xi$ is a formal power series for $\ln \Xi$ with terms given by \emph{clusters}.
A cluster $\Gamma$ is an ordered tuple of polymers, and to each $\Gamma$ we associate an \emph{incompatibility graph} $H(\Gamma)$. This graph has vertex set $\Gamma$ (i.e.\ a vertex for each polymer present in $\Gamma$ with multiplicity) and each pair of vertices representing incompatible polymers is an edge of $H(\Gamma)$. By convention, a polymer is incompatible with itself.
We write $\C$ for the set of all clusters.
The formal power series corresponding to the cluster expansion is then
\[ \ln \Xi = \sum_{\Gamma \in \C}\phi(\Gamma)\prod_{\gam\in\Gamma}w_\gam, \]
where $\phi$ is the Ursell function 
\[ \phi(\Gamma) = \sum_{\substack{F\subset E(H(\Gamma)) \\ (V(H(\Gamma)), F) \text{ connected}}} (-1)^{|F|}. \]
We remark that for any polymer $\gamma$, the tuples $(\gamma),(\gamma, \gamma),\dotsc$ are valid clusters, and so the cluster expansion is an infinite sum. 
We use the following theorem to establish convergence and an important tail bound.

\begin{theorem}[Koteck\'y--Preiss~\cite{KP86}]\label{thm:KP}
If there exist functions $f : \cP \to [0, \infty)$ and $g : \cP \to [0, \infty)$ such that for all $\gam \in \cP$, 
\begin{equation}
\label{eqn:KPcondition}
\sum_{\gamma' \not\sim \gamma}w_{\gamma'}e^{f(\gamma') + g(\gamma')} \leq f(\gamma),
\end{equation}
    then the cluster expansion converges absolutely. 
    Moreover, if we let $g(\Gamma)=\sum_{\gamma\in \Gamma}g(\gamma)$, then for every polymer $\gam$,
\begin{equation}
\label{eqn:KPtailbound}
\sum_{\substack{\Gamma \in \mathcal{C}\\ \Gamma \not\sim \gamma}}\left|\phi(\Gamma)\prod_{\gamma' \in \Gamma}w_{\gamma'}\right| e^{g(\Gamma)} \leq f(\gamma),
\end{equation}
where we write $\Gamma \not\sim \gamma$ to mean that the cluster $\Gamma$ contains a polymer that is incompatible with $\gamma$.
\end{theorem}

\subsection{Algorithms from convergent cluster expansions}\label{sec:clusteralgorithms}

Approximate counting and sampling algorithms for polymer models follow from Theorem~\ref{thm:KP} in a rather standard way due to~\cite{HPR19b} and subsequent works. 
In this paper we consider two polymer models that we define for an $n$-vertex, $d$-regular graph $G$, and then we consider approximating $\Xi$ or approximately sampling from $\nu$ with time complexities measured in terms of $n$ and a desired approximation error $\delta$.

For a polymer model as above, let $g:\cP\to[0,\infty)$ be as in Theorem~\ref{thm:KP}, and extend $g$ to clusters $\Gamma\in\C$ by $g(\Gamma) = \sum_{\gamma\in\Gamma}g(\gam)$.
An FPTAS for $\Xi$ follows from a few steps that one must check can be done efficiently:
\begin{enumerate}
    \item Use~\eqref{eqn:KPtailbound} to show that there exists some $L= L(\eps)$ such that 
    \[ \Xi(L):=\exp\left(\sum_{\substack{\Gamma \in \C\\g(\Gamma) \le L}}\phi(\Gamma)\prod_{\gam\in\Gamma}w_\gam\right) \] 
    is a close enough relative approximation of $\Xi$.
    \item List all clusters $\Gamma\in\C$ such that $g(\Gamma) \le L$.
    \item For each such cluster $\Gamma$, compute $\phi(\Gamma)$ and $\prod_{\gamma\in \Gamma} w_\gamma$.
    \item Compute $\Xi(L)$ directly from the above quantities.
\end{enumerate}
The running time of this algorithm depends on $L$, the definition of ``close enough,'' the time complexity of listing clusters and of computing polymer weights, $\phi$ and $g$. 
It has been established~\cite{HPR19b,JKP20,BCH+20} that for the polymer models we consider, the above steps can be carried out sufficiently fast for the definition of an FPTAS\@. 

Provided that the above algorithm yields an FPTAS for $\Xi$, there is a generic polynomial-time approximate sampling algorithm for $\nu$ that follows from a \emph{self-reducibility} property of the polymer model; see~\cite[Theorem~10]{HPR19b}. We omit the details and summarize these facts as the following theorem which applies when $\cP$ is (a subset of) the connected subgraphs of some bounded-degree graph $G$. For a subgraph $\gam$, we use $v_\gam$ to denote the number of vertices in $\gam$.
The important fact is that both the high-temperature and low-temperature polymer models we study satisfy the hypotheses of the result below.

\begin{theorem}[{\cite[Theorem~8]{JKP20} and~\cite[Theorem~10]{HPR19b}}]\label{thm:algorithm}
Fix $d$ and let $\mathcal{G}$ be some class of graphs of maximum degree at most $d$. 
Suppose the following hold for a polymer model with partition function $\Xi(G)$ where the polymers $\cP$ are connected subgraphs of some $G\in\mathcal{G}$, and with decay function $g$.
\begin{enumerate}
\item There exist constants $c_1, c_2 > 0$ such that given a connected subgraph $\gamma$, determining whether $\gamma$ is a polymer in $\mathcal{P}$, and computing $w_\gamma$ and $g(\gamma)$ can be done in time $O(v_{\gamma}^{c_1}e^{c_2v_{\gamma}})$.
\item There exists $\rho = \rho(d) > 0$ so that for every $G\in\mathcal{G}$ and every $\gamma\in\mathcal{P}$, $g(\gamma)\geq \rho v_{\gamma}$.
\item The Koteck\'y--Preiss condition~\eqref{eqn:KPcondition} holds with the given function $g$.
\end{enumerate}

Then there is an FPTAS for the partition function $\Xi(G)$ of the polymer model for $G\in\mathcal{G}$ that with approximation error $\delta$ and $n=|V(G)|$ has running time $n\cdot (n/\delta)^{O((\ln d+c_2)/\rho)}$.
There is also a polynomial-time sampling scheme for the associated polymer measure $\nu$.
\end{theorem}

\section{Low-Temperature Algorithms}\label{sec:lowtemp}

\subsection{Sketch of low-temperature argument}

For the low-temperature regime, we decompose the partition function into colorings that have a majority color, and the remaining colorings. 
Our polymer model can only represent colorings with a majority color, so it is vital to prove that the remaining colorings have a small contribution to the partition function, which is the content of our Lemma~\ref{lem:maingroundstate}. 
Working with fewer colors and smaller expansion makes this step more challenging than in~\cite{JKP20}. The core of the proof involves obtaining upper bounds on the number of vertex-colorings of a graph that induce few monochromatic edges. 
In light of the heuristic~\ref{enumproblem-colorings}, which relates colorings with few monochromatic edges to cuts, this is done by analyzing a Karger-like random sequence of edge contractions and the probability that this yields a particular cut.

We then handle each majority color with a polymer model where the polymers are connected, induced subgraphs of $G$ whose vertices do not receive the majority color. The weight of a polymer is related to the $(q-1)$-spin partition function on that subgraph. Since we are working with colorings with a majority color, the ``defects'' (vertices that don't have the majority color) occupy at most half the vertices, which is a global constraint on the polymers that we wish to avoid. 
We pass to a different polymer model which requires only that each polymer has at most half the vertices. This partition function is clearly an overestimate; however, in Lemma~\ref{lem:polymerapprox} we show that this overestimation is negligible.

All that remains is to verify the Koteck{\'y}--Preiss condition for this polymer model. 
There is, however, another challenge here: the aforementioned $(q-1)$-spin partition function plays a significant role in this computation, and one needs to bound it in order to proceed. Fortunately, our understanding of extremal bounds on such functions has progressed recently, most notably due to Sah, Sawhney, Stoner, and Zhao~\cite{SSSZ20}. We use their results to derive bounds on the partition functions of the subgraphs that make up our polymers. 
For triangle-free graphs the available bounds are stronger, which lets us work with weaker expansion in this case, and ultimately with the hypercube.
The extremal results we derive from~\cite{SSSZ20} are the content of Lemmas~\ref{lem:extremal} and~\ref{lem:TFextremal}.

The final ingredient in verifying the Koteck{\'y}--Preiss condition is a bound on the number of connected subsets of vertices in $G$ with a given edge-boundary size. At first glance, this looks similar to a lemma used in~\cite{JP20,JPP22} that gives a bound on the number of connected subsets of vertices with a given \emph{vertex-boundary} size (in bipartite graphs), which followed from mild adaptations of existing container methods~\cite{Sap87,Sap01}.

While the case of edge-boundaries does not seem to follow directly from these methods, we have developed Theorem~\ref{thm:enum} inspired by these previous works. At a high-level, container methods work by finding an ``encoding'' or ``certificate'' for each object of the type we wish to enumerate. 
If these certificates are small then the number of certificates, and thus objects, must be relatively small. 
In our case, we group sets of vertices $A$ with a fixed edge-boundary $B$ of size $b$  by identifying a ``core'' subset $A_0\subset A$ of size $O(b/d)$.
We show these core sets must be connected in $G^7$, and hence their number can be controlled. 
In this sense, each core set $A_0$ certifies a container which contains sets $A$ (of the type whose number we wish to bound) obtained by growing $A_0$ appropriately, see~\eqref{eqn:cover}.
To complete the bound on the number of sets $A$ of interest, we must bound for each core $A_0$ the number of possible $A$, or equivalently edge-boundaries $B$, that can be associated with $A_0$.
We achieve this using Karger's algorithm involving randomized edge-contractions to count cuts. 

\subsection{Polymer model}
In Theorem~\ref{thm:mainpotts}, we deal with the class of $d$-regular $2$-expanders and triangle-free $d$-regular $1$-expanders. Both these families of graphs satisfy
\begin{equation}\label{eqn:mincut2}
\text{every set}~A \subset V~\text{such that}~2\leq |A| \leq |V|/2~\text{satisfies}~|\nabla(A)| \geq 2d-2.
\end{equation}
In particular, this also holds for the hypercube $Q_d$. 
Property~\eqref{eqn:mincut2} is immediate in $d$-regular 2-expanders and follows from Mantel's theorem~\cite{Mantel} in the triangle-free 1-expanding case. 
To see this, note that the expansion condition suffices for $|A|\ge 2d-2$, and if $|A|<2d-2$ then we use the bound \[ \nabla(A) = d|A| - 2|E(G[A])| \ge d|A| - \frac{|A|^2}{2}, \]
which follows from Mantel's theorem. 
For $2\le |A| \le 2d-2$, this lower bound is at least $2d-2$ and property~\eqref{eqn:mincut2} follows.

The min-cut of a graph is the minimum number of edges that need to be deleted in order to disconnect the remaining graph. Any $d$-regular graph $G$ that satisfies~\eqref{eqn:mincut2} also satisfies
\begin{equation}\label{eqn:mincut1}
\text{the min-cut of}~G~\text{has size at least}~d.
\end{equation}
In other words, the smallest set of edges that can be deleted to disconnect the graph is simply the set of all edges incident to any given vertex. Throughout this section we will use that $G$ satisfies~\eqref{eqn:mincut2}, and therefore~\eqref{eqn:mincut1}.

We now define a polymer model for the low-temperature regime, following~\cite{JKP20}. 
Here we take $\eps\in(0,1)$; then, for some $d_0(\eps)$ and an absolute constant $c$, we have $d\ge d_0$, $q\ge d^c$, and $\beta \ge (1+\eps)\beta_o(q,d)$. 
We also assume that $q$ is large enough (by ensuring that $d_0$ and $c$ are large enough) that this implies $\beta \ge (2+\eps)\ln(q)/d$.

Let a polymer $\gam$ be a connected, induced subgraph of $G$ with at most $n/2$ vertices. For each polymer $\gam$, let $v_\gam$ and $e_\gam$ be the number of vertices and edges respectively. 
Let $\nabla_{\gamma}$ denote the number of boundary edges (i.e.\ edges of $G$ with one endpoint in the vertex set of $\gamma$), and let $\eta_{\gamma} := \nabla_{\gamma}/v_\gamma$, be a measure of the expansion of the vertex set of $\gam$ in $G$. 
The polymer weights $w_\gam$ are given by
\[
w_{\gamma} = e^{-\beta (\nabla_{\gamma} + e_{\gamma})} \cdot Z_{\gamma}(q-1, \beta),
\]
where we interpret $\gam$ as a graph and hence $Z_{\gamma}(q-1, \beta)$ is the Potts partition function of $\gam$ with $q-1$ colors. 
We say that two polymers $\gam$ and $\gam'$ are incompatible if $\dist_G(\gam, \gam') \leq 1$. 
That is, $\gam\not\sim\gam'$ if and only if they share a vertex or one contains a neighbor of the other. 
As before, we write $\Omega$ for the collection of pairwise compatible sets of polymers and $\Xi = \Xi_{G}(q,\beta)$ for the partition function of this polymer model.

For our low-temperature algorithm, there is not a precise correspondence between $Z_G(q,\beta)$ and $\Xi$, and we must establish that approximating $\Xi$ does in fact yield an approximation of $Z_G(q,\beta)$.
The set $[q]^{V(G)}$ of $q$-colorings of $V(G)$ admits a partition into $S_0, S_1, \dotsc, S_q$ where $S_0$ is the set of colorings such that each color occupies at most $n/2$ vertices and for $j\in [q]$, $S_j$ is the set of colorings in which $j$ is the majority color (i.e. $S_j = \{\sigma\in [q]^{V(G)} : |\sigma^{-1}(j)| > n/2\}$). We have the following lemma whose proof is postponed to Section~\ref{sec:maingroundstate}.

\begin{lemma}\label{lem:maingroundstate}
There exists an absolute constant $c$ such that the following holds. 
For every $\epsilon\in(0,1)$ there exists $d_0(\epsilon)$ such that for every $d\ge d_0$, $q\ge d^c$ and $\beta \ge (1+\eps)\beta_o(q,d)$, we have
\[
\frac{\sum_{\sigma \in S_0}e^{\beta m(G,\sigma)}}{Z_{G}(q, \beta)} \leq q^{-\Omega\left(\frac{\epsilon n}{d}\right)}.
\]
\end{lemma}

Let us define the partition function-like quantity
\[
\widetilde{\Xi} = \widetilde{\Xi}_G(q,\beta) := \sum_{\substack{\Lambda \in \Omega \\ \|\Lambda\| < n/2}}\prod_{\gamma \in \Lambda}w_{\gamma},
\]
where $\|\Lambda\| = \sum_{\gamma \in \Lambda}v_{\gamma}$ is the size of a set $\Lambda$ of pairwise compatible polymers. We write $\widetilde\Omega = \{\Lambda\in\Omega : \|\Lambda\| < n/2\}$ and note that $\widetilde\Xi$ is not the true partition function of the polymer model because of the global constraint $\|\Lambda\| < n/2$. 
In terms of the Potts partition function $Z_G(q,\beta)$, the quantity $\widetilde{\Xi}$ faithfully represents the colorings of $G$ which have a majority color because
\[
    \sum_{j\in[q]}\sum_{\sigma\in S_j}e^{\beta m(G,\sigma)} = e^{\beta dn/2} \sum_{j\in[q]}\sum_{\substack{\Lambda \in \Omega \\ \|\Lambda\| < n/2}} \prod_{\gam\in\Lambda}w_\gam = q e^{\beta dn/2}\widetilde{\Xi}.
\]
To see this, consider a coloring $\sigma\in[q]^{V}$ with majority color $j$ and decompose $\sigma^{-1}([q]\setminus\{j\})\subset V$ into sets such that their pairwise distance in $G$ is at least two. 
These sets induce subgraphs which form pairwise compatible polymers. Hence for any $j\in[q]$ the colorings with majority color $j$ are in bijection with $\{\Lambda\in\Omega : \|\Lambda\| < n/2\}$. 
The weight function and incompatibility criterion have been carefully defined to make the above calculation work on the level of partition function contribution as well.

Rewriting Lemma~\ref{lem:maingroundstate} in this terminology gives us that 
\[
\frac{q e^{\beta dn/2}}{Z_G(q,\beta)} \widetilde{\Xi} = 1 - q^{-\Omega\left(\frac{\epsilon n}{d}\right)}.
\]
Thus, approximating $\widetilde\Xi$ is a viable avenue for approximating $Z_G(q,\beta)$. 
There is one final complication, however, related to the fact that $\widetilde\Xi$ is not precisely the partition function $\Xi$ of the low-temperature polymer model defined above. 
Since the standard algorithmic setup (Theorem~\ref{thm:algorithm}) allows us to approximate $\Xi$ and approximately sample from the associated polymer measure $\nu$, we must additionally show that restricting to $\Lambda$ such that $\|\Lambda\|<n/2$ does not harm the approximation too much.
The proof of the low-temperature case of  Theorem~\ref{thm:mainpotts} therefore follows from the following two lemmas and the standard algorithmic setup of Theorem~\ref{thm:algorithm}.

\begin{lemma}\label{lem:polymerapprox}
For $\Xi$ and $\widetilde{\Xi}$ as defined above, we have
\[
\widetilde{\Xi} \le \Xi \le \widetilde\Xi\left(1 + e^{-\Omega(n/d)}\right).
\]
\end{lemma}

\begin{lemma}\label{lem:polymerpf}
For every $\epsilon\in(0,1)$, there exist $d_0(\epsilon)$ and an absolute constant $c$ such that for every $d\ge d_0$, $q\ge d^c$ and $\beta \ge (1+\eps)\beta_o(q,d)$, the Koteck\'y--Preiss condition (see~\eqref{eqn:KPcondition} in Theorem~\ref{thm:KP}) holds for the low-temperature polymer model with $f(\gam)=\frac{\eps v_\gam \ln q}{4d}$ and $g(\gam)=\frac{\eps \nabla_\gam \ln q}{4d} + \frac{v_\gam}{d}$.
\end{lemma}

Note that the left-hand side of condition~\eqref{eqn:KPcondition} decreases if $g$ decreases, and hence the above lemma also implies that Theorem~\ref{thm:KP} holds for any smaller choice of $g$.
We prove the above results, starting with Lemma~\ref{lem:polymerpf}, in the following subsections.

\subsection{Proof of Lemma~\ref{lem:polymerpf}}

We first give the proof in the case that $G$ is a 2-expander, and then show the modifications necessary for triangle-free 1-expanders.

We use $N(\gam) = N_G(\gam) = \{v: \exists u \in \gam \text{ such that }\{u,v\} \in E\}$ for the \emph{neighborhood} of $\gam$ in $G$ and $N_G[\gam] = N_G(\gam) \cup V(\gam)$ for the \emph{closed neighborhood} of $\gam$. We upper bound the left-hand side of condition~\eqref{eqn:KPcondition}, here noting that the incompatibility relation on polymers yields
\[ \sum_{\gam'\not\sim\gam}w_{\gam'} e^{f(\gam')+g(\gam')} \le  \sum_{u\in N_G[\gam]}\sum_{\gam' \ni u} w_{\gam'} e^{f(\gam') + g(\gam')} = \sum_{u\in N_G[\gam]}\sum_{b\ge d}\sum_{\substack{\gam' \ni u \\ \text{s.t.} \nabla_{\gam'}=b}} w_{\gam'} e^{f(\gam') + g(\gam')}, \]
Since $|N_G[\gam]| \le v_\gam + \nabla_\gam$, it suffices to show that for all $u\in V$ we have 
\begin{equation}
\label{eqn:KPsplit}
\sum_{b\ge d}\sum_{\substack{\gam' \ni u \\ \text{s.t.} \nabla_{\gam'}=b}} w_{\gam'} e^{f(\gam') + g(\gam')} \leq \min_{\gam\in\cP} \frac{f(\gam)}{v_{\gam} + \nabla_{\gam}}.
\end{equation}

We use the following lemma, whose proof is postponed to Section~\ref{sec:extremal}, to bound the weights $w_{\gam'}$ that appear in~\eqref{eqn:KPsplit}.

\begin{lemma}\label{lem:extremal}
For any polymer $\gam\in\cP$ we have
\[
Z_{\gam}(q-1, \beta) \leq e^{\beta e_\gam} q^{v_\gam - \frac{2e_\gam}{d}+\frac{e_\gam}{\binom{d+1}{2}}} 2^{\frac{e_\gam}{\binom{d+1}{2}}}.
\]
\end{lemma}

To bound the exponent of $q$ in Lemma~\ref{lem:extremal}, observe that $d v_\gam = 2e_\gam + \nabla_\gam \ge 2e_\gam$ by a double-counting argument, and recall that $\eta_{\gamma} = \nabla_{\gamma}/v_\gamma$ to obtain
\[
v_\gam - \frac{2e_\gam}{d}+\frac{e_\gam}{\binom{d+1}{2}} \leq \frac{\nabla_\gam}{d}\left(1 + \frac{1}{\eta_\gam}\right).
\]
Plugging these facts and Lemma~\ref{lem:extremal} into the definition of $w_\gam$ gives
\[
w_\gam = e^{-\beta(e_\gam +\nabla_\gam)} Z_\gam(q-1,\beta) \leq e^{-\beta \nabla_\gam} q^{\frac{\nabla_\gam}{d}(1 + \frac{1}{\eta_\gam})}2^{\frac{\nabla_\gam}{d\eta_\gam}},
\]
and using $\beta \ge (2+\eps)\ln(q)/ d$ we have
\[
e^{-\beta \nabla_{\gamma}} \leq e^{-(2+\eps)\nabla_\gam \frac{\ln q}{d}} = q^{-(2+\eps)\frac{\nabla_\gam}{d}}.
\]
We assume $q$ is large enough (by assuming $d_0$ and $c$ are large enough) so that $2^{\frac{\nabla_\gam}{d\eta_\gam}}\leq q^{\frac{\nabla_\gam}{d} \cdot \frac{\eps}{4}}$, and thus
\[
w_\gam \leq q^{-\frac{\nabla_\gam}{d}\left(1+\frac{3\eps}{4} - \frac{1}{\eta_\gam}\right)}.
\]
Using our choice of $f$ and $g$, we also have
\[
e^{f(\gam)+g(\gam)}
= q^{\frac{\eps}{4d}(v_\gam+\nabla_\gam)}\cdot e^{v_\gam/d} =  q^{\frac{\nabla_\gam}{d}\frac{\eps}{4}\left(1+\frac{1}{\eta_\gam}\right)}\cdot e^{\frac{v_\gam}{d}},
\]
which combined with the bound $\eta_\gam\ge 2$ gives
\begin{equation}
\label{eqn:weightupperbound}
w_\gam e^{f(\gam)+g(\gam)} \leq q^{-\frac{\nabla_\gam}{d}\left(1+\frac{3\eps}{4} - \frac{1}{\eta_\gam} - \frac{\eps}{4\eta_\gam}-\frac{\eps}{4}\right)}e^{\frac{v_\gam}{d}} \leq q^{-\frac{\nabla_\gam}{d}\left(\frac{1}{2}+\frac{3\eps}{8}\right)}e^{\frac{v_\gam}{d}} \leq q^{-\frac{\nabla_\gam}{2d}},
\end{equation}
where the final inequality holds again when $q$ is large enough in terms of $\eps$.

We now use our container-type result. Applying Theorem~\ref{thm:enum} and~\eqref{eqn:weightupperbound} to bound the left-hand side of~\eqref{eqn:KPsplit}, we have (for $q$ large enough that the series converges)
\[ \sum_{b\ge d}\sum_{\substack{\gam' \ni u \\ \text{s.t.} \nabla_{\gam'}=b}} w_{\gam'} e^{f(\gam') + g(\gam')} \leq \sum_{b\ge d} d^{O\left((1+1/\eta)b/d\right)} q^{-\frac{b}{2d}} \le \frac{d^{c'}/\sqrt q}{1-(d^{c'}/\sqrt q)^{1/d}},\]
where the absolute constant $c'\ge 4$ is large enough that the term $O\left((1+1/\eta)b/d\right)$ is at most $c'b/d$ for all $b\ge d$ and all $\eta\ge 1$. 
Note that the right-hand side above is a decreasing function of $q$.
To satisfy~\eqref{eqn:KPsplit} we want this to be at most
\[ \min_{\gam\in\cP} \frac{f(\gam)}{v_{\gam} + \nabla_{\gam}} = \frac{\eps \ln q}{4d} \min_{\gam\in\cP} \frac{v_\gam}{v_\gam+\nabla_\gam} = \frac{\eps \ln q}{4d(d+1)}, \]
which is an increasing function of $q$. 
Thus, it suffices to take $q$ larger than some known value for the desired inequality holds. 
It is straightforward to check that $q=d^{3c'}$ suffices for all $d$ large enough.

\subsection{The triangle-free case}
If $G$ is triangle free then a stronger upper bound on partition functions of the form $Z_{\gam}(q-1, \beta)$ holds. 

\begin{lemma}\label{lem:TFextremal}
If $G$ is triangle free, then
\[
Z_{\gam}(q-1, \beta) \leq e^{\beta e_\gam} \cdot q^{v_\gam - \frac{2e_\gam}{d}+\frac{e_\gam}{d^2}} \cdot 2^{\frac{e_\gam}{d^2}} .
\]
\end{lemma}

Plugging this inequality in the computation in the previous section, we have
\[ v_\gam - \frac{2e_\gam}{d} + \frac{e_\gam}{d^2} \le \frac{\nabla_\gam}{d}\left(1+\frac{1}{2\eta_\gam}\right),\]
and hence only assuming $\eta\ge 1$ we have the bound
\[
w_\gam e^{f(\gam)+g(\gam)} \leq q^{-\frac{\nabla_\gam}{d}\left(\frac{1}{2}+\frac{\eps}{4}\right)}e^{\frac{v_\gam}{d}} \leq q^{-\frac{\nabla_\gam}{2d}} .
\]
We may then apply Theorem~\ref{thm:enum} to reach the same conclusion as before.

\subsection{Proof of Lemma~\ref{lem:polymerapprox}}

Let $\nu$ and $\widetilde{\nu}$ denote the probability distributions on $\Omega$ given by the partition functions $\Xi$ and $\widetilde{\Xi}$ respectively, meaning that 
$\nu(\Lambda) = \prod_{\gam\in\Lambda}w_\gam / \Xi$ for any $\Lambda\in\Omega$ and $\widetilde{\nu}(\Lambda)=0$ if $\|\Lambda\|\ge n/2$ but $\widetilde{\nu}(\Lambda) = \prod_{\gam\in\Lambda}w_\gam / \widetilde{\Xi}$ otherwise.
We briefly prove a large deviation inequality on the size of the defects, similar to~\cite{JP20, JPP22}. 
Here, we carefully make use of the term $e^{v_\gam/d}$ in the definition of $g(\gam)$ and the monotonicity of condition~\eqref{eqn:KPcondition} in $g$.
This immediately implies that Theorem~\ref{thm:KP} holds for an alternative polymer model on the same polymers but with weights
\[
w'_{\gamma} : = w_{\gamma} \cdot e^{\frac{v_{\gamma}}{d}},
\]
because for the same choice $f(\gam)=\frac{\eps v_\gam\ln q}{4d}$ and now taking $g(\gam) = \frac{\eps\nabla_\gam\ln q}{4d}$, condition~\eqref{eqn:KPcondition} is the same for both models.

Write
\[
\Xi' : = \sum_{\Lambda \in \Omega }\prod_{\gamma \in \Lambda}w'_{\gamma} = \sum_{\Lambda \in \Omega }\prod_{\gamma \in \Lambda}w_{\gamma} e^{v_\gam/d},
\]
and let $\mathbf{\Lambda}$ denote a random set of pairwise compatible polymers drawn from $\nu$.  
Observe that $\Xi'$ naturally appears in an expression for $\mathbb{E} e^{\|\mathbf{\Lambda}\|/d}$ as follows:
\[ \mathbb{E} e^{\|\mathbf{\Lambda}\|/d} = \sum_{\Lambda\in\Omega}\nu(\Lambda)e^{\|\Lambda\|/d} = \frac{1}{\Xi}\sum_{\Lambda\in\Omega}\prod_{\gam\in\Lambda}\left(w_\gam e^{v_\gam/d}\right) = \frac{\Xi'}{\Xi}. \]

Then Theorem~\ref{thm:KP} and the fact that $\Xi \ge 1$ gives us
\[
\ln \mathbb{E} e^{\|\mathbf{\Lambda}\|/d} = \ln \frac{\Xi'}{\Xi} \le \ln \Xi' = \sum_{\Gamma\in\mathcal{C}}\phi(\Gamma)\prod_{\gam\in\Gamma} w'_\gam, \] 
where the final equality is the cluster expansion for the abstract polymer model defined by $\Xi'$, and hence the sum is over clusters $\Gamma\in\mathcal{C}$ as defined in Section~\ref{sec:prelim}.
We can crudely upper bound this cluster expansion of by summing over all polymers of the form $\gam=(\{v\},\emptyset)$ induced by a vertex and applying the tail bound~\eqref{eqn:KPtailbound}.
This gives
\begin{align}
\ln\Xi' 
&\leq \sum_{v\in V(G)}\sum_{\substack{\Gamma \in \mathcal{C}\\ \Gamma\not\sim (\{v\},\emptyset)}} \left|\phi(\Gamma)\prod_{\gamma \in \Gamma}w'_{\gamma}\right| 
\\&\leq \sum_{v\in V(G)}\sum_{\substack{\Gamma \in \mathcal{C}\\ \Gamma\not\sim (\{v\},\emptyset)}} q^{-\eps/4}\left|\phi(\Gamma)\prod_{\gamma \in \Gamma}w'_{\gamma}\right|e^{g(\Gamma)}\label{eq:crudetail}
\\&\leq n \cdot \frac{\epsilon \ln q}{4q^{\epsilon/4}d}.\label{eq:crudetailwithKP}
\end{align}
Line~\eqref{eq:crudetail} holds because every cluster $\Gamma\not\sim(\{v\},\emptyset)$ contains at least one polymer (which has at least one vertex) and hence using~\eqref{eqn:mincut1} we have
\[ g(\Gamma) = \frac{\eps\ln q}{4d}\sum_{\gamma\in\Gamma}\nabla_\gamma \ge \frac{\eps\ln q}{4}.  \]
To obtain line~\eqref{eq:crudetailwithKP} we use~\eqref{eqn:KPtailbound} with $f((\{v\},\emptyset)) = \frac{\eps\ln q}{4d}$. 

Combining the above calculations, we have $\mathbb{E} e^{\|\mathbf{\Lambda}\|/d} \le e^{\frac{n}{4d}}$, and hence Markov's inequality gives
\begin{equation}\label{eqn:finallargedeviation}
\mathbb{P}(\|\mathbf{\Lambda}\| \geq n/2) = \mathbb{P}(e^{\|\mathbf{\Lambda}\|/d} \ge e^{n/(2d)}) \leq \exp\left(\frac{\epsilon n \ln q}{4q^{\epsilon/4}d} - \frac{n}{2d}\right) = \exp\left(-\frac{n}{4d}\right),
\end{equation}
which implies $\nu(\widetilde\Omega)\ge 1-e^{-n/(4d)}$. Here we have used the fact that $q$ is large in terms of $\epsilon$ so $\frac{\epsilon \ln q}{4q^{\epsilon/4}} \leq \frac{1}{4}$. We would like to point out that similar computations are required for the proofs of Theorems~\ref{thm:pt} and~\ref{thm:maincube}.

Since the probabilities for $\widetilde\nu$ on outcomes in $\widetilde\Omega$ are made by redistributing the probability mass $\nu(\Omega\setminus \widetilde\Omega)\le e^{-n/(4d)}$ onto outcomes in $\widetilde\Omega$, we immediately have a bound on the total variation distance
\[
\|\widetilde{\nu} - \nu\|_{TV} \leq e^{-n/(4d)}.
\]
The conclusion follows because by the definitions of $\widetilde\nu$ and total variation distance we have 
\[ 0 \le \widetilde\nu(\widetilde\Omega)-\nu(\widetilde\Omega) = 1-\nu(\widetilde\Omega) \le \|\widetilde{\nu} - \nu\|_{TV}, \]
and hence dividing by $\nu(\widetilde\Omega)$ gives
\[
0 \le \frac{1}{\nu(\widetilde\Omega)}-1 = \frac{\Xi}{\widetilde\Xi}-1 \le \frac{\|\widetilde{\nu} - \nu\|_{TV}}{\nu(\widetilde\Omega)} \le \frac{e^{-n/(4d)}}{1-e^{-n/(4d)}}.
\]

\section{Enumerative results}

In this section, we prove Theorem~\ref{thm:enum} and Lemma~\ref{lem:maingroundstate}.

\subsection{Connected sets with small edge boundaries: Proof of Theorem~\ref{thm:enum}}\label{sec:enum}

The main combinatorial result underlying our container-like lemma is the following standard adaptation of Karger's algorithm to count \emph{$\alpha$-min-cuts}. 

\begin{theorem}[Karger~\cite{Kar93}]\label{thm:karger}
Let $G$ be a graph whose min-cut has $t$ edges, and let $\alpha \geq 1$ be a real number. The number of cuts in $G$ with at most $\alpha t$ edges is at most $\binom{n}{2\alpha} \cdot 2^{2\alpha}$.
\end{theorem}

We also require a well-known result giving an upper bound to the number of connected induced subgraphs containing a fixed vertex. 

\begin{proposition}[\cite{KNUTH98}, Vol.~3 p396, Ex.11]\label{prop:trees}
The number of rooted labelled trees of maximum degree $\Delta$ on $n \geq 1$ vertices is 
    \[
    \frac{\binom{\Delta n}{n}}{(\Delta-1)n + 1} \leq (e\Delta)^{n-1},
    \]
    and hence, for a graph of maximum degree $\Delta$ and a fixed vertex $v$, the number of connected induced subgraphs on $n$ vertices containing $v$ is also at most this number.
\end{proposition}
With this, we now prove Theorem~\ref{thm:enum}.

\begin{proof}[Proof of Theorem~\ref{thm:enum}]
We first introduce some notation. Recall $V$ and $E$ are the sets of vertices and edges of $G$ respectively. For $u, v \in V$, we use $\dist_G(u, v)$ to denote the \emph{distance} in $G$ from $u$ to $v$, that is, the number of edges on a shortest path between $u$ and $v$. For every $k \in \mathbb{N}$, we define the \emph{$k$th neighborhood} of $u$ in $G$ as $N_{G}^k(u) = \{v \in V : 1 \leq \dist_G(u, v) \leq k\}$. Then $N_G(u) = N_G^1(u)$. We may also define the $k$th neighborhood of a set of vertices $A \subseteq V$ as $N_G^k(A) = \cup_{u \in A}N_G^k(u)$. The \emph{$k$th power of the graph $G$} is the graph $G^k$ on vertex set $V$ where $\{u,v\} \in E(G^k)$ if and only if $1 \leq \dist_G(u, v) \leq k$. Observe that $N_{G^k}(u) = N_G^k(u)$.

Let $I$ be the \emph{vertex-edge incidence graph} of $G$, which is a bipartite graph with vertex set $V \sqcup E$ and edges $\left\{(u,e) \in V\times E\mid u \in e\right\}$.
Let $J$ be the bipartite graph with vertex set $V \sqcup E$ and edges given by $\left\{(u,e)\in V\times E\mid\exists v\in N_G(u) \text{ s.t. } v\in e\right\}$.
That is, $(u,e)$ forms an edge of $J$ if and only if $e$ is an edge incident in $G$ to a neighbor of $u$ (this includes all edges incident to $u$ itself). 
Note that $I$ is a subgraph of $J$.

Observe that for every $u\in V$, the degree of $u$ in $J$ is at least $\binom{d+1}{2} \ge d^2/2$ and at most $d^2$.
To see the lower bound, consider that each of the $d$ neighbors of $u$ is itself incident to $d$ edges, and the number of distinct edges incident to $N_G(u)$ is minimized when $N_G(u)$ forms a clique on $d+1$ vertices. The upper bound holds because the maximum number of edges of $G$ incident to $N_G(u)$ is $d^2$ (which occurs if and only if $u$ is contained in no cycles of length 3 or 4).

Now consider a set of vertices $A \subseteq V$ of size at most $n/2$ such that $G[A]$ (the subgraph of $G$ induced by $A$) is connected and $x \in A$. Let $B=\nabla(A)$ and $|B|=|\nabla(A)|=b$. 

We then have $N_I(A)= B\sqcup W$ where $W=E(G[A])$.
Note that by summing degrees, $d|A|=b + 2|W|$, and by the expansion assumption~\eqref{eqn:expansion} we have $b \ge \eta|A|$. Hence,
\begin{equation}
\label{eqn:Wbound}
|W|\le \frac{b}{2}\left(\frac{d}{\eta}-1\right) \le \frac{bd}{2\eta}.
\end{equation}
Observe that $I$ is a subgraph of $J$, so $N_I(A) \subseteq N_J(A)$. Let $B' = N_J(A) \setminus N_I(A)$. If $e \in B'$, then $e$ must contain the external endpoint of some boundary edge of $A$. There are at most $b$ such external endpoints, and each is incident to at most $d-1$ such non-boundary edges. Thus, \begin{equation}
\label{eqn:B'bound}
|B'| \leq b(d-1).
\end{equation}

Let $A_0 \subseteq A$ be a maximal subset of vertices with pairwise-disjoint neighborhoods in $J$. We will first determine how many choices there are for $A_0$. We claim that 
\begin{equation}\label{eqn:size}
|A_0| \le \frac{2}{d^2}\left(|B'| +|B| +|W|\right) \le \frac{b}{d}\left(2+\frac{1}{\eta}\right).
\end{equation}
The first inequality holds because the above definitions give $|N_J(A)| = |B'| + |B| + |W|$, and each vertex in $A$ has degree at least $d^2/2$ in $J$.
The second inequality follows from~\eqref{eqn:Wbound} and~\eqref{eqn:B'bound}.

The maximality of $A_0$ gives us that 
\begin{equation}\label{eqn:span} 
\text{for any}~u \in A \setminus A_0,~\text{there is a}~v \in A_0~\text{such that}~N_{J}(u) \cap N_{J}(v) \neq \emptyset.
\end{equation}
As a result, we have, 
\begin{equation}
\label{eqn:cover}
A_0 \cup N^2_{J}(A_0) \supseteq A.
\end{equation}

Moreover, we make the following claim (where we recall $x \in A$ is some specified vertex).

\begin{claim}\label{claim:connected}
The set of vertices $A_0 \cup \{x\}$ is connected in $G^7$.
\end{claim}

\begin{proof}
We start by proving that $A_0$ is connected in $G^7$. 
Since $A$ is connected in $G$, and $A_0\subseteq A$, for any pair of distinct vertices $u,v\in A_0$, there must be a path $P$ from $u$ to $v$ along edges of $G$ using only vertices in $A$. 
Let the vertices of $P$ be $w_0 = u, w_1, \dots, w_k = v$.

Observe that every vertex in $P$ is distance at most three in $G$ from $A_0$. This follows from~\eqref{eqn:span}; indeed, either $w_i \in A_0$ or there is some $u_i \in A_0$ and $e \in G$ such that $e \in N_J(w_i) \cap N_J(u_i)$. This means $e$ contains a vertex in $N_G(w_i)$ and a vertex in $N_G(u_i)$. There are three possibilities for $e$: either $e = u_iv_i$, $e$ is incident to exactly one of $u_i$ or $w_i$, or $e$ is incident to neither $u_i$ nor $w_i$. These cases give $\dist_G(w_i, u_i)$ as 1, 2, and 3, respectively. See Figure~\ref{fig:dist} for an illustration.

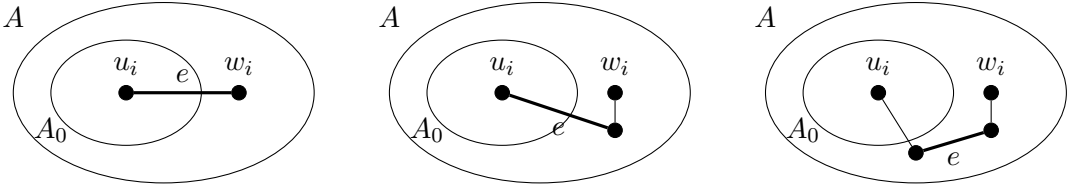
\begin{figure}
    \centering

	\begin{tikzpicture}[vert/.style ={circle,fill=black,draw,minimum size=0.5em,inner sep=0pt}]
	\draw[shift={(-5,0)}] (-1,0) ellipse (1cm and 0.7cm);
	\draw[shift={(-5,0)}] (-0.5,0) ellipse (2.2cm and 1.4cm);
	\node[shift={(-5,0)}, vert, label=above:$u_i$] (A) at (-1,0) {};
	\node[shift={(-5,0)}, vert, label=above:$w_i$] (B) at (.5,0) {};
	\draw[shift={(-5,0)}, very thick] (A) -- (B) node [midway, above] {$e$};
	\node[shift={(-5,0)}] () at (-2,-0.5) {$A_0$};
	\node[shift={(-5,0)}] () at (-2.5,1) {$A$};

	\draw (-1,0) ellipse (1cm and 0.7cm);
	\draw (-0.5,0) ellipse (2.2cm and 1.4cm);
	\node[vert, label=above:$u_i$] (C) at (-1,0) {};
	\node[vert] (D) at (.5,-0.5) {};
	\node[vert, label=above:$w_i$] (E) at (.5,0) {};
	\draw[very thick] (C) -- (D) node [midway, below] {$e$};
	\draw (D) -- (E);
	\node () at (-2,-0.5) {$A_0$};
	\node () at (-2.5,1) {$A$};

	\draw[shift={(5,0)}] (-1,0) ellipse (1cm and 0.7cm);
	\draw[shift={(5,0)}] (-0.5,0) ellipse (2.2cm and 1.4cm);
	\node[shift={(5,0)}, vert, label=above:$u_i$] (F) at (-1,0) {};
	\node[shift={(5,0)}, vert] (G) at (-.5,-0.8) {};
	\node[shift={(5,0)}, vert] (H) at (.5,-0.5) {};
	\node[shift={(5,0)}, vert, label=above:$w_i$] (I) at (.5,0) {};
	\draw[shift={(5,0)}] (F) -- (G);
	\draw[shift={(5,0)}, very thick] (G) -- (H) node [midway, below] {$e$};
	\draw[shift={(5,0)}] (H) -- (I);
	\node[shift={(5,0)}] () at (-2,-0.5) {$A_0$};
	\node[shift={(5,0)}] () at (-2.5,1) {$A$};
	\end{tikzpicture}
    \caption{The three different cases for an edge $e$ to be in $N_J(u_i) \cap N_J(w_i)$\label{fig:dist}}
\end{figure}

We can now construct a walk $P'$ along edges of $G$ from $u$ to $v$ that visits a vertex in $A_0$ at most every seven steps. Let $Q_i$ be a path from $w_i$ to $A_0$ such that $|Q_i| \leq 3$, and let $Q_i^{-1}$ be the reverse path from $A_0$ to $w_i$. If $w_i \in A_0$, set $Q_i$ and $Q_i^{-1}$ to be empty. Then the edges of $P'$ are
\[P' = (uw_1, Q_1, Q_1^{-1}, w_1w_2, Q_2, Q_2^{-1}, \dots, w_{k-2}w_{k-1}, Q_{k-1}, Q_{k-1}^{-1}, w_{k-1}w_k)\]
\begin{figure}
	\centering

	\begin{tikzpicture}[vert/.style ={circle,fill=black,draw,minimum size=0.5em,inner sep=0pt}]
	\node[vert, label=below:$w_i$] (w1) at (0,0) {};
	\node[vert, label=below:$w_{i+1}$] (w2) at (2,0) {};
	\node[vert, label=below:$w_{i+2}$] (w3) at (4,0) {};
	\node[vert] (A1) at (0,1) {};
	\node[vert] (B1) at (0,2) {};
	\node[vert] (C1) at (0,3) {};
	\draw (w1) -- (A1) -- (B1) -- (C1);
	\draw[dashed] (w1) -- (-1,0);
	\draw[dashed] (w3) -- (5,0);
	\draw[->, color=red] (-0.2,0.2) -- (-0.2, 0.8);
	\draw[->, shift={(0,1)}, color=red] (-0.2,0.2) -- (-0.2, 0.8);
	\draw[->, shift={(0,2)}, color=red] (-0.2,0.2) -- (-0.2, 0.8);
	\draw[->, shift={(0,2)}, color=red] (0.2, 0.8) -- (0.2,0.2);
	\draw[->, shift={(0,1)}, color=red] (0.2, 0.8) -- (0.2,0.2);
	\draw[->, color=red] (0.2, 0.8) -- (0.2,0.2);
	\draw (w1) -- (w2) -- (w3);
	\draw[->, color=red] (0.7,0.2) -- (1.3,0.2);	

	\node[vert, shift={(2,0)}] (A2) at (0,1) {};
	\node[vert, shift={(2,0)}] (B2) at (0,2) {};

	\draw[shift={(2,0)}] (w2) -- (A2) -- (B2);
	\draw[->, shift={(2,0)}, color=red] (-0.2,0.2) -- (-0.2, 0.8);
	\draw[->, shift={(2,1)}, color=red] (-0.2,0.2) -- (-0.2, 0.8);

	\draw[->, shift={(2,1)}, color=red] (0.2, 0.8) -- (0.2,0.2);
	\draw[->, shift={(2,0)}, color=red] (0.2, 0.8) -- (0.2,0.2);
	\draw[->, shift={(2,0)}, color=red] (0.7,0.2) -- (1.3,0.2);

	\draw [dashed] plot [smooth] coordinates {(-0.3, 3) (0.5,2.2) (1,1.9) (2.5,1.8)} node[label={above right:$A_0$}]{};
	\end{tikzpicture}
	\caption{An example of part of the path $P'$ in red arrows.\label{fig:path}}
\end{figure}
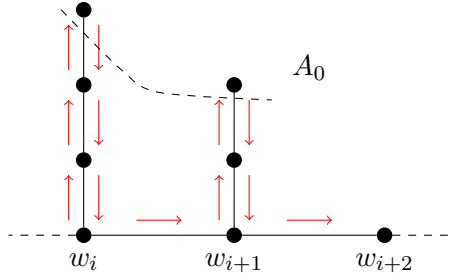

See Figure~\ref{fig:path} for an example. The walk $P'$ visits vertices of $A_0$ at most every seven steps, so there is a walk in $G^7$ from $u$ to $v$ using only vertices in $A_0$. As $u$ and $v$ were arbitrary, it follows that $A_0$ is connected in $G^7$.

This completes the proof of the claim in the case $x \in A_0$. If $x \notin A_0$, then it is distance at most three in $G$ from the nearest vertex in $A_0$ and hence is adjacent to a vertex of $A_0$ in $G^7$. Thus, $A_0\cup \{x\}$ is connected in $G^7$ as required.
\end{proof}

As a result, we can specify $A_0$ by specifying a tree of degree at most $d^7$ rooted at $v$ (since $G^7$ has maximum degree $d(d-1)^6\le d^7$). Using Proposition~\ref{prop:trees} and the bound on $|A_0|$ given by~\eqref{eqn:size}, there are at most $(ed^7)^{\frac{b}{d}\left(2+\frac{1}{\eta}\right)} = d^{O((1+1/\eta)b/d)}$ possibilities for $A_0$. 

We now want to count the number of choices for $B$. Let $E_0$ be the edges of $G$ in $N^3_J(A_0)$, meaning the edges which have an endpoint at distance at most 5 from $A_0$, and let $G'$ be the graph $G$ with every edge in $E \setminus E_0$ contracted. By~\eqref{eqn:cover} we have that 
\[
B \subseteq N_{I}(A) \subseteq N_{J}(A) \subseteq N_{J}^3(A_0),
\]
so $B \subset E(G')$. We also have $v \in V(G')$ if and only if $\dist_G(u, v) \leq 5$ for some $u \in A_0$. Since $|N_G^{k-1}(A_0)| \leq (d+1)|N_G^k(A_0)|$ for $k \geq 1$, this implies $|V(G')| \leq (d+1)^4|A_0| = O(bd^3)$.

By~\eqref{eqn:mincut1}, the min-cut size of $G$ is at least $d$ (and in fact exactly $d$ since $G$ is $d$-regular). Contracting edges does not decrease the min-cut size of a graph, so the min-cut size of $G'$ is also at least $d$.

Thus, the number of possible boundaries $B$ is upper bounded by the number of cuts of size at most $b$ in $G'$, which by Theorem~\ref{thm:karger} is at most
\[
\binom{O(bd^3)}{\leq 2b/d} \cdot 2^{2b/d} = d^{O(b/d)}.
\]
We are done, as we can uniquely determine $A$ from its boundary.
\end{proof}

\subsection{Colorings with small color classes: Proof of Lemma~\ref{lem:maingroundstate}}\label{sec:maingroundstate}

Let us now turn to Lemma~\ref{lem:maingroundstate}. For a coloring $\sigma$, let us use $\nm(\sigma) = \nm(G,\sigma)$ to denote the number of non-monochromatic edges in $\sigma$; observe that $m(G,\sigma) = |E(G)| - \nm(\sigma)$. Recall that $S_0$ is the set of colorings in which each color occupies at most $n/2$ vertices. The following lemma is the main point of this section.

\begin{lemma}\label{lem:groundstate}
Let $G$ be as in Theorem~\ref{thm:mainpotts}. For $k \geq d$, let $\mathcal{L}_k$ denote the number of $q$-colorings from $S$ which give exactly $k$ non-monochromatic edges. Then for $\delta > \frac{12}{\ln d}$,
\[
\mathcal{L}_k \leq n^4 \cdot q^{\frac{(2+\delta)k}{d}-\Omega\left(\frac{\delta k}{d}\right)}.
\]
\end{lemma}

We would like to briefly comment on parameter $\delta$. If one is interested in a combinatorial bound, one could simply plug in the ``best possible'' value (i.e., $\delta = 12/\ln d$). We allow for a choice in the parameter $\delta$ solely to make the subsequent calculations slightly easier.

Before using this to prove Lemma~\ref{lem:maingroundstate}, we will need the following.

\begin{lemma}\label{lem:nmedges}
    Let $G$ be an $\eta$-expander on $n$ vertices. For every $\sigma\in S_0$ we have $\nm(\sigma) \ge \eta n /2$.
\end{lemma}
\begin{proof}
    Let $A_i=\sigma^{-1}(i)$ be the set of vertices which get color $i$ under $\sigma$, which by the fact that $\sigma\in S_0$ satisfies $|A_i|\le n/2$. The non-monochromatic edges of $G$ under $\sigma$ are precisely those that appear as $\nabla(A_i)$ for two distinct choices of $i\in [q]$.
    This gives
    \[ 2\nm(\sigma) = \sum_{i \in[q]}|\nabla (A_i)| \ge \eta\sum_{i\in [q]}|A_i| = \eta n. \qedhere \]
\end{proof}

We are now ready to prove Lemma~\ref{lem:maingroundstate}.

\begin{proof}[Proof of Lemma~\ref{lem:maingroundstate}]
We start off with the easy lower bound $Z_{G}(q,\beta) \geq q \cdot e^{\beta |E(G)|}$
obtained by considering only the colorings that give every vertex the same color. So for $\beta \geq (2+\epsilon)\frac{\ln q}{d}$, we have
\begin{align*}
\frac{\sum\limits_{\sigma \in S_0}e^{\beta\left(|E(G)| - \operatorname{nm}(\sigma)\right)}}{Z_{G}(q,\beta)} &\le \frac{\sum\limits_{k \geq \eta n/2}\sum\limits_{\substack{\sigma \in S_0 \\ \operatorname{nm}(\sigma) = k}}e^{\beta(|E(G)| - k)}}{q \cdot e^{\beta |E(G)|}} \\
& \leq \frac{1}{q}\sum_{k \geq \eta n/2}\sum_{\substack{\sigma \in S_0 \\ \operatorname{nm}(\sigma) = k}}q^{-(2+\epsilon)\frac{k}{d}} \\
& \leq \frac{n^4}{q}\sum_{k \geq \eta n/2}q^{-\Omega\left(\frac{\epsilon k}{d}\right)}
\end{align*}
where the first inequality uses Lemma~\ref{lem:nmedges} and the last inequality uses Lemma~\ref{lem:groundstate} with $\delta = \epsilon$ (note that we already assume that $d \geq \exp\left(\Omega(1/\epsilon)\right)$). If $d \geq \sqrt{n}$, then $\frac{n^4}{q} < 1$. Otherwise, $n^4 = q^{o\left(\frac{\epsilon k}{d}\right)}$. 
In either case,
\[
\frac{n^4}{q}\sum_{k \geq \eta n/2}q^{-\Omega\left(\frac{\epsilon k}{d}\right)} \leq q^{-\Omega\left(\frac{\epsilon n}{d}\right)}.\qedhere
\]
\end{proof}

 One ingredient in the proof of Lemma~\ref{lem:groundstate} is the following, which requires the same techniques as in the proof of Theorem~\ref{thm:karger}.

\begin{lemma}\label{lem:numbercolors}
Let $G$ be as in Theorem~\ref{thm:mainpotts}. Let $\mathcal{C}(\ell, s)$ be the number of $q$-colorings of the vertices of $G$ such that there are exactly $\ell$ non-monochromatic edges, where each monochromatic component has at least $s$ vertices. Then we have
\begin{enumerate}
\item\label{lem:numbercolors1} $\mathcal{C}(\ell,1) \leq \binom{n}{2\ell/d}q^{\frac{2\ell}{d}}$, and
\item\label{lem:numbercolors2} $\mathcal{C}(\ell,2) \leq \binom{n}{2\ell/d}q^{\frac{2\ell}{2d - 2}}\cdot (2d)^{\frac{2\ell}{d}}$.
\end{enumerate}
\end{lemma}

\begin{proof}
For part~\ref{lem:numbercolors1}: Fix any $q$-coloring $\sigma$ with exactly $\ell$ non-monochromatic edges and run the following algorithm on $G$. 

\begin{itemize}
\item[I.] While there are more than $\frac{2\ell}{d}$ vertices, choose a uniformly random edge and contract it. Delete self-loops, if any, after each contraction. 
\item[II.] Color every vertex uniformly and independently from $[q]$. Output the final graph $G'$ and coloring $\sigma'$.
\end{itemize}

Observe that $\sigma'$ naturally corresponds to a $q$-coloring $\sigma$ in the original graph defined by $\sigma(v) = \sigma'(v')$ where $v'$ is the vertex in $G'$ into which $v$ was contracted.

Recall that $G$ has a min-cut of size $d$ by~\eqref{eqn:mincut1}, and it does not decrease with the contraction operation during step I. So at any point in this step, if the current graph has $m$ vertices, then it has at least $\frac{md}{2}$ edges.

Moreover, each contraction reduces the number of vertices by $1$, and all deleted self-loops must be monochromatic. So the probability that the $\ell$ non-monochromatic edges remain uncontracted in $G'$ is at least
\[
\left(1-\frac{2\ell}{nd}\right)\left(1 - \frac{2\ell}{(n-1)d}\right)\cdots \frac{d}{2\ell} = \binom{n}{2\ell/d}^{-1}.
\]

During step II, the probability that every contracted vertex gets the correct color is $q^{-\frac{2\ell}{d}}$.
Therefore, the probability that $\sigma$ was recovered by this procedure is at least $\binom{n}{2\ell/d}^{-1}q^{-\frac{2\ell}{d}}$. Since this holds for any $\sigma$, there are at most $\binom{n}{2\ell/d}q^{\frac{2\ell}{d}}$ many $q$-colorings with $\ell$ non-monochromatic edges.

For part~\ref{lem:numbercolors2}: Fix any $q$-coloring $\sigma$ with at least $\ell$ non-monochromatic edges where each monochromatic component has at least $2$ vertices, and consider the following algorithm. 
\begin{itemize}
\item[I.] While there are more than $\frac{2\ell}{d}$ vertices, choose a uniformly random edge and contract it. Delete self-loops, if any, after each contraction. 
\item[II.] While a vertex in this contracted graph has degree $d$, contract it with one of its neighbors uniformly at random, deleting self-loops at each stage.
\item[III.] While there are more than $\frac{2\ell}{2d - 2}$ vertices, choose a uniformly random edge and contract it. Delete self-loops, if any, after each contraction. 
\item[IV.] Color every vertex uniformly and independently at random from $[q]$.
\end{itemize}

Again, the final coloring naturally corresponds to a $q$-coloring in $G$. Step I is analyzed in an identical manner as $(1)$. The probability that the $\ell$ non-monochromatic edges remain uncontracted at the end of this step is at least $\binom{n}{2\ell/d}^{-1}$.

Since every monochromatic component in $\sigma$ has at least two vertices, every uncontracted vertex at the end of step I is incident to at least one monochromatic edge. So during step II, the probability that the $\ell$ non-monochromatic edges remain uncontracted is at most $d^{-\frac{2\ell}{d}}$.

After the end of step II,~\eqref{eqn:mincut2} guarantees that the min-cut is at least $2d-2$, and this does not decrease with the contraction operation in step III\@. So, using a calculation similar to that in the analysis of step I, the probability that the $\ell$ non-monochromatic edges remain uncontracted at the end of this step is at least $2^{-\frac{2\ell}{d}}$.

During step IV, the probability that every contracted vertex gets the correct color is $q^{-\frac{2\ell}{2d-2}}$.
Thus, $\sigma$ is recovered with probability at least $\binom{n}{2\ell/d}^{-1}q^{-\frac{2\ell}{(2d-2)}}(2d)^{-\frac{2\ell}{d}}$. Since this holds for any $\sigma$, we have that there are at most $\binom{n}{2\ell/d}q^{\frac{2\ell}{(2d-2)}}(2d)^{\frac{2\ell}{d}}$ many $q$-colorings with $\ell$ non-monochromatic edges and each monochromatic component having size at least $2$.
\end{proof}

One can check by hand that our bound is not tight in several cases, such as when $s = 1$ and $\ell = d$. Indeed, the original algorithm by Karger for min-cuts gives a reasonable tight lower bound on the number of near-minimum cuts when $G$ is a cycle---a case that is not included in the class of graphs under consideration for Theorem~\ref{thm:mainpotts}. However, recent work obtaining a more fine-grained understanding of the Karger process may be adaptable to our setting of $q$-colorings; we refer the interested reader to \cite{GHLL21}.

We are now ready to prove Lemma~\ref{lem:groundstate}. Recall that $\mathcal{L}_k$ is the number of $q$-colorings from $S_0$ which induce exactly $k$ non-monochromatic edges, and we want to show that
\[\mathcal{L}_k \leq n^4 \cdot q^{\frac{(2+\delta)k}{d}-\Omega\left(\frac{\delta k}{d}\right)}. \]

\begin{proof}[Proof of Lemma~\ref{lem:groundstate}]
\textbf{Case 1: $k \leq nd^{1 - \frac{\delta}{4}}$.} For a subset $A \subset V$, let $E(A)$ denote the set of edges induced by $A$. Consider a partition of the vertices of $G$ into $L$ monochromatic components. Let $T$ be the set of components containing exactly one vertex. Let $N$ be the set of non-monochromatic edges. We then have the partition
\[
N = E(T) \sqcup \nabla(T) \sqcup \left(N \cap E\left(\overline{T}\right)\right).
\]
Let $e_T := |E(T)|$, $b_T := |\nabla(T)|$, and $\ell := \left|N\cap E (\overline{T})\right| = k - e_T - b_T$. We have that $|T| \leq \frac{2k}{d} \leq nd^{-\delta} < n/3$ and each monochromatic component has size at most $n/2$. So, one can choose some $S \subseteq V$ such that $n/3 \leq |S|\leq 2n/3$ by greedily including all vertices from the smallest remaining monochromatic component. In particular, $S \supseteq T$. The expansion condition~\eqref{eqn:expansion} then implies that $|\nabla(S)| \geq n/3$. However, $\nabla(S) \subseteq \nabla(T) \cup (N \cap E(\overline{T}))$, and so $|\nabla(S)| \leq b_T + (k - e_T - b_T) = k - e_T = \ell + b_T$. As a result, we have
\begin{equation}\label{eqn:largeedges}
b_T + \ell = \Omega(n).
\end{equation}

A $q$-coloring counted by $\mathcal{L}_k$ can be chosen by (I) first choosing the cut $\left(T, \overline{T}\right)$, (II) giving each element on the $T$-side of the cut a color, and (III) choosing the monochromatic components in $V \setminus T$ so that each component has at least two vertices. 

Theorem~\ref{thm:karger} gives us that (I) can be done in at most $\binom{n}{2b_T/d} \cdot 2^{2b_T/d}$ ways. We also have that (II) can be done in at most $q^{|T|}$ ways. Part~\ref{lem:numbercolors2} of Lemma~\ref{lem:numbercolors} gives us that (III) can be done in at most $\binom{n}{2\ell/d}\cdot q^{\frac{2\ell}{2d-2}}\cdot (2d)^{\frac{2\ell}{d}}$ ways. 
So, the number of $q$-colorings of $V$ with the above parameters is at most
\begin{align*}
Q & :=\binom{n}{2b_T/d} \cdot 2^{2b_T/d} \cdot q^{|T|} \cdot \binom{n}{2\ell/d}\cdot q^{\frac{2\ell}{2d-2}} \cdot (2d)^{\frac{2\ell}{d}} \\
& = \binom{n}{2b_T/d}\cdot \binom{n}{2\ell/d}\cdot q^{\frac{\ell + 2e_T + b_T}{d}} \cdot q^{\frac{\ell}{d^2}} \cdot (2d)^{\frac{2\ell}{d}}, \\
\end{align*}

where we have used $d|T| = 2e_T + b_T$. So we have 
\begin{align*}
Q \cdot q^{-\frac{2k}{d}} & = \binom{n}{2\ell/d}\binom{n}{2b_T/d}\cdot q^{\frac{-b_T - \ell}{d}} \cdot q^{\frac{\ell}{d^2}}\cdot (2d)^{\frac{2\ell}{d}} \\
& \leq \binom{n}{2(\ell + b_T)/d} \cdot 2^{\frac{4(b_T + \ell)}{d}} q^{\frac{-b_T - \ell}{d}} \cdot q^{\frac{\ell}{d^2}} \cdot (2d)^{\frac{2\ell}{d}} \\
& \leq d^{O\left(\frac{\ell+ b_T}{d}\right)}q^{-\frac{\ell + b_T}{d}} \cdot 2^{\frac{4(b_T + \ell)}{d}}q^{\frac{\ell}{d^2}} \\
& \leq q^{-\Omega\left(\frac{n}{d}\right)},
\end{align*}
where the first inequality used the fact that for $a+b \leq x$, we have $\binom{x}{a}\binom{x}{b} \leq \binom{x}{a+b}4^{a+b}$, and the the next two inequalities use~\eqref{eqn:largeedges}. So summing over all values for $\ell$ and $b_T$ gives us that for $k \leq nd^{1 - \frac{\delta}{4}}$,
\begin{align*}
\mathcal{L}_k \leq n^4 \cdot Q \leq n^4 \cdot q^{\frac{2k}{d} - \Omega\left(\frac{n}{d}\right)}.
\end{align*}

\textbf{Case 2: $k \geq nd^{1 - \frac{\delta}{4}}$.} This case is relatively straightforward. Using part~\ref{lem:numbercolors1} of Lemma~\ref{lem:numbercolors}, we have 
\[
\mathcal{L}_k \leq \binom{n}{2k/d}q^{\frac{2k}{d}} \leq (ed)^{\frac{2\delta k}{3d}} \cdot q^{\frac{2 k}{d}} = q^{-\frac{(2 + \delta)k}{d} - \Omega\left(\frac{\delta k}{d}\right)}.\qedhere
\]
\end{proof}

\section{Extremal results for partition functions}\label{sec:extremal}

An important step in upper bounding the weight $w_\gamma = e^{-\beta(\nabla_\gamma + e_\gamma)}Z_{\gamma}(q-1,\beta)$ is to upper bound the partition function $Z_{\gamma}(q-1,\beta)$. 
Bounds on partition functions in this setting are a well-studied topic in extremal combinatorics.

\subsection{General graphs and cliques}

\begin{theorem}[{Sah, Sawhney, Stoner, Zhao~\cite[Theorem 1.14 restated]{SSSZ20}}]\label{thm:extremal}
    Let $G$ be a graph, $q\ge 2$ be an integer and $\beta\ge 0$. Then
    \[ Z_G(q,\beta) \le \prod_{v\in V(G)} Z_{K_{d_v+1}}(q,\beta)^{\frac{1}{d_v+1}}, \]
    where $d_v$ is the degree of a vertex $v\in V(G)$.
\end{theorem}

We briefly remark that their result was much more general; \cite{SSSZ20} shows that in fact every ferromagnetic model (using a definition introduced in~\cite{GSV16} as a generalization of the ferromagnetic Potts model) is ``clique-maximizing'' which translates to the above statement on the level of partition functions.

\begin{corollary}\label{cor:extremal}
    Let $G$ be a graph of maximum degree $\Delta$ with $n$ vertices and $m$ edges. Then for any integer $q\ge 2$ and $\beta\ge 0$,
    \[ Z_G(q,\beta) \le q^{n-\frac{2m}{\Delta}} Z_{K_{\Delta+1}}(q,\beta)^{\frac{2m}{\Delta(\Delta+1)}} . \]
\end{corollary}
\begin{proof}
    This follows from Theorem~\ref{thm:extremal} and the fact that $Z_{K_{d+1}}(q,\beta)^{\frac{1}{d+1}}$ forms a log-convex sequence for $d\ge 0$, which is proved in~\cite[Lemma~5.3]{SSSZ20}. 
    In particular, consider the graph with partition function $Z_G(q,\beta)^\Delta$ formed from the disjoint union of $\Delta$ copies of $G$. 
    Applying Theorem~\ref{thm:extremal}, we obtain an upper bound on $Z_G(q,\beta)^\Delta$ in terms of a degree sequence. 
    We can further increase the upper bound because the aforementioned log-convexity implies that when $a\le b$, replacing two degrees $a$ and $b$ by degrees $a-1$ and $b+1$ can only increase the bound. 
    Then over graphs of maximum degree $\Delta$ with $\Delta n$ vertices and $\Delta m$ edges, the maximum upper bound Theorem~\ref{thm:extremal} can give is from a degree sequence where $2m$ vertices have degree $\Delta$ and the rest have degree zero. 
    The result follows upon taking the power $1/\Delta$ of this inequality.
\end{proof}

\subsection{Triangle-free graphs and bicliques}

The upper bound in Theorem~\ref{thm:extremal} can be improved in the case that $G$ is triangle-free. 

\begin{theorem}[{Sah, Sawhney, Stoner, Zhao~\cite[Theorem 1.9 and the subsequent remark]{SSSZ20}}]\label{thm:TFextremal}
    Let $G$ be a triangle-free graph with no isolated vertices, $q\ge 2$ be an integer and $\beta\ge 0$. Then
    \[ Z_G(q,\beta) \le \prod_{uv\in E(G)} Z_{K_{d_u,d_v}}(q,\beta)^{\frac{1}{d_u d_v}}, \]
    where $d_u$ is the degree of a vertex $u\in V(G)$.
\end{theorem}

\begin{corollary}\label{cor:TFextremal}
    Let $G$ be a triangle-free graph of maximum degree $\Delta$ with $n$ vertices and $m$ edges. Then for any integer $q\ge 2$ and $\beta\ge 0$,
    \[ Z_G(q,\beta) \le q^{n-\frac{2m}{\Delta}} Z_{K_{\Delta,\Delta}}(q,\beta)^{\frac{m}{\Delta^2}} . \]
\end{corollary}
\begin{proof}
    Let $G$ have $t$ isolated vertices. Then by Theorem~\ref{thm:TFextremal} we have 
    \[ Z_G(q,\beta) \le q^t \prod_{uv\in E(G)} Z_{K_{d_u,d_v}}(q,\beta)^{\frac{1}{d_u d_v}}. \]
    The desired result is now a consequence of the fact that for $1\le a\le c$ and $1\le b\le d$ we have 
    \begin{equation}\label{eq:monotonicity} 
        q^{-\left(\frac{1}{a}+\frac{1}{b}\right)} Z_{K_{a, b}}(q,\beta)^{\frac{1}{a b}} \le q^{-\left(\frac{1}{c}+\frac{1}{d}\right)} Z_{K_{c,d}}(q,\beta)^{\frac{1}{c b}}.
    \end{equation}
    To see this, apply the inequality to each term of the product and collect the factors of $q$ to obtain
    \[ Z_G(q,\beta) \le q^{t + \sum_{uv\in E(G)}\left(\frac{1}{d_u}+\frac{1}{d_v}\right) - \frac{2m}{\Delta}}Z_{K_{\Delta,\Delta}}(q,\beta)^{\frac{m}{\Delta^2}}.  \]
    Some simple counting gives
    \[ \sum_{uv\in E(G)}\left(\frac{1}{d_u}+\frac{1}{d_v}\right) = n - t \]
    since each vertex $u$ appears as the endpoint of precisely $d_u$ edges, and the result follows.

    It remains to prove inequality~\eqref{eq:monotonicity}, which follows from H\"older's inequality and Jensen's inequality as noted\footnote{See Definition 1.11 and the subsequent discussion of monotonicity in~\cite{SSSZ20}. Note that for this monotonicity to hold one must work with probability measures on the sets of spins, which is not necessarily the case for other ``scale-free'' results proved in~\cite{SSSZ20} and related works.} in~\cite{SSSZ20}. 
    For concreteness, we observe that inequality~\eqref{eq:monotonicity} is a direct consequence of a generalized H\"older inequality stated as Theorem~3.1 in~\cite{LZ15}. 
    We can write $\Omega=[q]$, and let $\mu$ be the uniform probability measure on $\Omega$ and $\mu^a$ be the uniform probability measure on $\Omega^a$ such that for the non-negative functions $f:\Omega^2\to\mathbb{R}$ and $g:\Omega^a\to\mathbb{R}$ given by
    \begin{align*}
        f(x,y) &= \begin{cases}e^\beta & x = y \\ 1 & \text{otherwise},\end{cases} \\
        g(x_1,\dotsc,x_a) &= \int_{\Omega} \prod_{i=1}^a f(x_i,y) \,\mathrm d\mu.
    \end{align*}
    Then to prove~\eqref{eq:monotonicity} in this notation, we have
    \[ \mathrm{LHS} := q^{-\left(\frac{1}{a}+\frac{1}{b}\right)} Z_{K_{a, b}}(q,\beta)^{\frac{1}{a b}} = \left(\int_{[q]^a}|g|^b\,\mathrm d\mu^a\right)^{\frac{1}{ab}}. \]
    When $b\le d$ we can apply the generalized H\"older inequality (or in this somewhat special case iterated applications of the usual H\"older inequality followed by Jensen's inequality) to the integral over $[q]^a$ above to obtain the upper bound 
    \[ \mathrm{LHS} \le \left(\int_{\Omega^a} g^d \,\mathrm d\mu^a \right)^{\frac{1}{ad}} = q^{-\left(\frac{1}{a}+\frac{1}{d}\right)} Z_{K_{a, d}}(q,\beta)^{\frac{1}{a d}}. \]
    A symmetric application of this argument for any integer $c$ such that $a\le c$ gives inequality~\eqref{eq:monotonicity}.
\end{proof}

\begin{lemma}\label{lem:BGJ}
    Let $q \geq d^{4}$ and for some fixed $\epsilon > 0$ let $\beta \ge (1+\epsilon)\beta_o$. Then, 
    \begin{enumerate}
    \item $Z_{K_{d+1}}(q,\beta) \le (1+q^{-\Omega(\epsilon)})q e^{\beta\binom{d+1}{2}}$, and
    \item $Z_{K_{d,d}}(q,\beta) \le (1+ q^{-\Omega(\epsilon)})qe^{\beta d^2}$.
    \end{enumerate}
\end{lemma}

\begin{proof}
We will use the fact that both $K_{d+1}$ and $K_{d,d}$ are $\eta$-expanders with $\eta = d/2$. We will state the proof for $K_{d+1}$, and proof for $K_{d,d}$ follows similarly. Let $\sigma$ be any coloring of the vertices of $K_d$, and let $c(\sigma)$ denote the number of colors that appear in $\sigma$. 

We first observe that the proof of Lemma~\ref{lem:groundstate} also extends to all $\sigma$ such that $c(\sigma) \geq 2$. Indeed, the only place where $S$ was used was to establish Lemma~\ref{lem:nmedges}, and~\eqref{eqn:largeedges}. Both of these also hold for $K_{d+1}$ whenever there are at least two colors.
So we have
\begin{align*}
\frac{\sum_{\sigma:c(\sigma) \geq 2}e^{\beta\left(\binom{d+1}{2} - \operatorname{nm}(\sigma)\right)}}{q \cdot e^{\beta \binom{d+1}{2}}} & \leq \frac{1}{q}\sum_{k \geq d}\sum_{\substack{\sigma:c(\sigma) \geq 2 \\ \operatorname{nm}(\sigma) = k}}q^{-(2+\epsilon)\frac{k}{d}} \\
& \leq \frac{d^4}{q}\sum_{k \geq d}q^{-\Omega\left(\frac{\epsilon k}{d}\right)} \\
& \leq q^{-\Omega\left(\epsilon\right)}.\qedhere
\end{align*}
\end{proof}

The proofs of Lemma~\ref{lem:extremal} and Lemma~\ref{lem:TFextremal} are now straightforward calculations.

\begin{proof}[Proof of Lemma~\ref{lem:extremal}]
By Corollary~\ref{cor:extremal} and Lemma~\ref{lem:BGJ},
\begin{align*} Z_{G[\gam]}(q-1, \beta) &\leq (q-1)^{v_\gam-\frac{2e_\gam}{d}} Z_{K_{d+1}}(q-1, \beta)^{\frac{2e_\gam}{d(d+1)}}\\
&\leq (q-1)^{v_\gam-\frac{2e_\gam}{d}}\left(2(q-1) e^{\beta \binom{d+1}{2}}\right)^{\frac{2e_\gam}{d(d+1)}}\\
&<q^{v_\gam - \frac{2e_\gam}{d}+\frac{e_\gam}{\binom{d+1}{2}}} e^{\beta e_\gam}2^{\frac{e_\gam}{\binom{d+1}{2}}}.\qedhere
\end{align*}

\end{proof}

\begin{proof}[Proof of Lemma~\ref{lem:TFextremal}]
By Corollary~\ref{cor:TFextremal} and Lemma~\ref{lem:BGJ},
\begin{align*} Z_{G[\gam]}(q-1, \beta)&\leq (q-1)^{v_\gam - \frac{2e_\gam}{d}}Z_{K_{d,d}}(q-1,\beta)^{\frac{e_\gam}{d^2}}\\
&\leq (q-1)^{v_\gam - \frac{2e_\gam}{d}} \cdot (2(q-1)e^{\beta d^2})^{\frac{e_\gam}{d^2}}\\
&\leq q^{v_\gam - \frac{2e_\gam}{d}+\frac{e_\gam}{d^2}} \cdot 2^{\frac{e_\gam}{d^2}} \cdot e^{\beta e_\gam}.\qedhere
\end{align*}

\end{proof}

\section{High-temperature algorithms}\label{sec:hightemp}

The polymer model we use in the high-temperature region is the same as the one considered in~\cite{BCH+20,CDK+20,HJP20}, which we now define.
Given a graph $G=(V,E)$, the set $\cP$ of polymers is the set of connected subgraphs (not necessarily induced) of $G$ on at least two vertices. 
The polymer weights are given by $w_\gamma = q^{1 - v_{\gamma}}p^{e_{\gamma}}$, where we write $p=e^\beta -1$. 
Two polymers $\gam$ and $\gam'$ are incompatible if and only if their union is connected in $G$.

The key observation is that edge subsets $A\subset E$ are in bijection with the collection $\Omega$ of sets of pairwise compatible polymers by the map taking $A$ to the set of components of $(V,A)$ on at least two vertices.
In this setup, the polymer model partition function $\Xi$ is given by 
\begin{align*} 
    \Xi &= \sum_{\Lambda\in\Omega}\prod_{\gam\in\Lambda}w_\gam 
    \\&= \sum_{A\subset E}\;\prod_{\substack{\text{components }\gamma \\ \text{of } (V,A)}}\;q^{1-v_\gam}p^{e_\gam} 
    \\&=\sum_{A\subset E}q^{c(A)-n}(e^\beta-1)^{|A|} = q^{-n} Z_G(q,\beta),
\end{align*}
where $c(A)$ is the number of components of the graph $(V,A)$.
This shows that $q^n \Xi$ is precisely the function which we wish to approximate, and hence any multiplicative approximation to $\Xi$ corresponds to a multiplicative approximation to $Z$ with the same relative error.
Given the standard approach to approximating polymer model partition functions outlined in Section~\ref{sec:clusteralgorithms}, to complete our high-temperature algorithm it suffices to verify the conditions of Theorem~\ref{thm:KP} which imply that the cluster expansion of this polymer model is convergent.

\begin{lemma}\label{lem:hightempKP}
    For every $\eps\in(0,1)$ there exist $d_0(\eps)$ and $q_0(\eps)$ such that for every $d\ge d_0$, $q\ge q_0$ and $\beta \le (1-\eps)\beta_o(q,d)$, the Koteck\'y--Preiss condition (see~\eqref{eqn:KPcondition} in Theorem~\ref{thm:KP}) holds for the high-temperature polymer model with $f(\gam)=g(\gam)=v_\gam$.
\end{lemma}

\begin{proof}
We assume that $d_0(\epsilon)$ and $q_0(\epsilon)$ are large enough that $\beta \leq (1 - \epsilon)\beta_o$ implies that $\beta \leq (2-\epsilon)\ln(q)/d$ and hence $q\ge (1+p)^{d/(2-\epsilon)}$.

Write $\cP_{u,k,t}$ for the subset of polymers $\gamma\in\cP$ such that $\gamma$ contains a fixed vertex $u$, $v_\gam=k$, and $e_\gam=t$. 
Given an arbitrary polymer $\gam$, we bound from above the sum in the left-hand side of~\eqref{eqn:KPcondition} as follows
\[
\sum_{\gamma' \not\sim \gamma}w_{\gamma'}e^{f(\gamma') + g(\gamma')}\leq \sum_{u \in \gamma}\sum_{\gamma' \ni u}w_{\gamma'}e^{f(\gamma') + g(\gamma')}.
\]
Given this and our choice $f(\gam) = v_\gam$, it suffices to show that for every vertex $u$ in $G$,
\[
F_u := \sum_{\gamma' \ni u}w_{\gamma'}e^{f(\gamma') + g(\gamma')} \leq 1.
\]

Observe that 
\begin{align*}
F_u &= q\sum_{\gamma' \ni u}q^{-v_{\gamma'}}p^{e_{\gamma'}} e^{f(\gamma') + g(\gamma')}
= q\sum_{k\ge 2} \sum_{t\ge k-1}|\cP_{u,k,t}|q^{-k}p^{t} e^{2k}.
\end{align*}
For $t < k-1$ we have $|\cP_{u,k,t}|=0$ since polymers must be connected, hence the above sums over $t$ start at $k-1$.
We bound $|\cP_{u,k,t}|$ from above to continue, noting that
\[ |\cP_{u,k,t}| \le \min\left\{ (ed)^{k-1} \cdot \binom{\binom{k}{2} - k + 1}{t - k + 1},\; (ed)^{k-1} \cdot \binom{dk/2 -k + 1}{t - k + 1}\right\}. \]
To see this, observe that there are at most $(ed)^{k-1}$ choices of a spanning tree of $\gam\in \cP_{u,k,t}$ by Proposition~\ref{prop:trees}, and the binomial coefficient gives an upper bound on the number of ways of adding edges of $G$ to this spanning tree to form $\gam$.
In the case $k \le d+1$, we can use $\binom{k}{2}$ as an upper bound on the number of edges of $G$ incident to a vertex of this spanning tree; else, we use $dk/2$. 
Then 
\begin{align*}
    F_u &\le 
    \frac{q}{ed}\sum_{k=2}^{d} (e^3d/q)^k p^{k-1} \sum_{t=k-1}^{\binom{k}{2}} p^{t-k+1}\binom{\binom{k}{2}-k+1}{t-k+1} 
    \\& \hphantom{\le {}} + \frac{q}{ed}\sum_{k=d+1}^{\infty} (e^3d/q)^k p^{k-1} \sum_{t=k-1}^{dk/2} p^{t-k+1}\binom{dk/2-k+1}{t-k+1}.
\end{align*}

The inner sums over $t$ are precisely $(1+p)^{\binom{k}{2}-k+1}$ and $(1+p)^{dk/2-k+1}$ respectively, so that 
\begin{align*}
    F_u &\le 
    \frac{1+p}{p}\frac{q}{ed}\left(\sum_{k=2}^{d} (e^3dp/q)^k (1+p)^{\binom{k}{2}-k}+ \sum_{k=d+1}^{\infty} (e^3dp/q)^k(1+p)^{dk/2-k}\right).
\end{align*}

Let $r = e^3 d p(1+p)^{d/2-1}/q$,
and observe that for $k\in [2,d]$ we have $\binom{k}{2}-k \le \frac{dk}{2}-k-\frac{d}{2}$. 
Then 
\begin{align*}
    F_u &\le 
    \frac{1+p}{p}\frac{q}{ed}\left(\sum_{k=2}^{d} \frac{r^k}{(1+p)^{d/2}}+ \sum_{k=d+1}^{\infty} r^k\right) = \frac{e^2}{1-r}\left(r-r^d +(1+p)^{d/2}r^d\right).
\end{align*}
To finish the proof, observe that in order to show $F_u\le 1$ it suffices to show that $(1+p)^{1/2}r \le 1/100$. 

Substituting the definition of $r$ and $q\ge (1+p)^{d/(2-\epsilon)}$ into $(1+p)^{1/2}r$ we have
\[ (1+p)^{1/2}r \le  e^3 d p(1+p)^{\frac{d}{2} - \frac{d}{2-\eps} - \frac{1}{2}}. \]
This has a unique stationary point on $p\in[0,\infty)$ at $p^*=\frac{4-2 \eps }{(d +1)\eps -2}$ provided that $d \ge 10/\eps$ (which also ensures that $p^*>0$), and this stationary point is a maximum. 
Assuming that $d\ge 10/\eps$, one can verify that $(1+p)^{1/2}r<1/100$ whenever $p \geq 100p^*/\epsilon^2$.
In terms of $q$, to ensure that $p\ge 100p^*/\epsilon^2$, it suffices to take $q > \exp(400/\eps^3)$.
This establishes Lemma~\ref{lem:hightempKP}.
\end{proof}

We briefly note that the above proof requires $\epsilon > 0$ so we cannot, for instance, simply allow $\epsilon < 0$ in order to extend Theorem~\ref{thm:mainpotts} to all temperatures. However, we can extend our results if we modify our high-temperature polymer model to consider only a subset of polymers---namely, those at most a certain size. We discuss this approach in slightly more detail in Section~\ref{sec:conclusion}.

We also note that Theorem~\ref{thm:FPTASuniqueness} follows from Lemma~\ref{lem:hightempKP} and the standard method sketched in Subsection~\ref{sec:clusteralgorithms} which uses Theorem~\ref{thm:algorithm}.

\begin{proof}[Proof of Theorem~\ref{thm:FPTASuniqueness}]
As both $q$ and $d$ tend to infinity, there exists $\beta_u(q,d)\sim \ln(q)/d$ such that the Gibbs uniqueness region of the infinite $d$-regular tree corresponds to $\beta<\beta_u$ (see~\cite{Hag96,BGP16}). 
Since $\beta_o\sim2\ln(q)/d$ we can pick (say) $\eps=1/4$ in Lemma~\ref{lem:hightempKP} and so long as $d$ and $q$ are large enough (i.e.\ at least some absolute constants) we will have the conditions of that lemma, as well as $\beta_u \le (1-\eps)\beta_o$. 
The algorithm is given in Subsection~\ref{sec:clusteralgorithms}.
\end{proof}

\section{Typical structure of the Potts model}\label{sec:pt}

In this section, we give a proof of Theorem~\ref{thm:pt}. A useful perspective on the random cluster model comes from tilted bond percolation, providing a distribution on edge subsets of a graph that with $q=1$ is precisely the usual notion of (independent) bond percolation. For integer $q$, there is a coupling between the distributions given by the Potts and random cluster models due to Edwards and Sokal~\cite{ES88} which we now describe. Given a subset of edges from the random cluster model, color each component with a color uniformly from $[q]$ to obtain a vertex coloring distributed according to the Potts model. 
In reverse, one performs bond percolation (with a particular probability) on the monochromatic components of a coloring from the Potts model.
See e.g.~\cite{Sok05} for a thorough treatment of the function $Z$ and the various combinatorial quantities that it encodes, including the equivalence of~\eqref{eqn:pottspartition} and~\eqref{eq:FKrep}. 

\begin{proof}[Proof of Theorem~\ref{thm:pt}] For part~\ref{part:pt1}, we use the Koteck\'y--Preiss theorem for the high-temperature cluster expansion to argue that all polymers have size $O(\ln n)$ with high probability. In order to do this, let $L=(1+\delta)\ln n$ and consider the truncated cluster expansion for $\Xi$. Recall that
\[
\Xi(L) = \exp\left(\sum_{\substack{\Gamma \in \mathcal{C}\\g(\Gamma) \leq L}}\phi(\Gamma)\prod_{\gamma \in \Gamma}w_{\gamma}\right).
\]

We have that $\left|\ln\Xi(L)-\ln\Xi\right|\leq n^{-\delta}$. Therefore, if we let $\tilde{Z}_G$ be the sum over configurations with components of size at most $L$, we have 
\[
\frac{\tilde{Z}_G}{Z_G} = \frac{q^n\Xi(L)}{q^n\Xi} \geq e^{-n^{-\delta}} = 1- \Omega(n^{-\delta}),
\]
and so, with high probability all components are on at most $(1+\delta)\ln n$ vertices. To get a configuration of the Potts model we make use of the Edwards–Sokal coupling \cite{ES88} by coloring each component at random with one of the $q$ colors.
Note that if $C_1,\dots,C_K$ are the components we have 
\[
\sum_{i=1}^K |C_i|^2 \leq \max_{i \in \{1,\ldots,K\}}|C_i| \cdot \sum_{i = 1}^K|C_i| \leq (1+\delta)n \ln n .
\]
Let $N_j$ be the number of vertices with color $j$. Clearly, $\mathbb{E}N_j = n/q$. By Hoeffding's inequality~\cite[Theorem $2$]{HOEFFDING63}, we have that 
\[
\mathbb{P}\left(\left|N_j - \frac{n}{q}\right| > tn \right) \leq  2 \exp\left(\frac{-2t^2 n^2}{\sum_{i=1}^K |C_i|^2}\right) \leq 2\exp\left(- \frac{2t^2 n}{(1+\delta)\ln n}\right).
\]
Taking $t=o\left(\frac{1}{q}\right)\cap \omega\left(\sqrt{\frac{\ln n \cdot \ln q}{n}}\right)$ and a union bound over all colors $q$, we have that with high probability all the $N_j$ satisfy $N_j=(1+o(1))n/q$.

The strengthening of Part~\ref{part:pt2} now follows. Replacing $n/2$ with $\delta n$ in~\eqref{eqn:finallargedeviation}, we have that
\begin{align*}
\mathbb{P}(\|\mathbf{\Lambda}\| \geq (1-\delta)n) &= \mathbb{P}(e^{\|\mathbf{\Lambda}\|/d} \ge e^{(1-\delta)n/d})\\
&\leq \exp\left(\frac{\epsilon n \ln q}{4q^{\epsilon/4}d} - (1-\delta)\frac{n}{d}\right) \\
\end{align*}

This gives us that for some $\delta = O\left(\frac{\epsilon\ln q}{q^{\epsilon/4}}\right)$ we have
\begin{equation}
\label{eqn:pottsprediction}
\text{the largest color class has at least }n\left(1 - \delta\right)~\text{vertices w.h.p.}
\end{equation}
\end{proof}

Note that there is some subtlety in the final argument. We did not start out by defining the ``majority colorings'' (those not in $S_0$) to be the ones with at least $(1-o(1))n$ vertices of the same color; if we had, the combinatorial arguments necessary to prove the Kote\'{c}ky-Preiss condition would have been far more difficult (or in some cases false) due to the weakness of our expansion assumptions. Instead, we take a two-layered approach. In the first layer, we use combinatorial arguments to show that $Z_G(q,\beta)$ is exponentially well-approximated by the majority colorings, where ``majority'' here means at least $\frac{n}2$ vertices of the same color (Lemma~\ref{lem:maingroundstate}). In the second layer, we show that this approximation is itself well-approximated by a polymer model with convergent cluster expansion (Lemma~\ref{lem:polymerapprox}). Due to the convergence (Lemma~\ref{lem:polymerpf}), we can discern an even finer level of detail within the polymer model---that the polymer model partition function is itself well-approximated by majority colorings, where we now take ``majority'' to mean at least $(1-o(1))n$ vertices of the same color. Our approximations at each step are good enough that we can chain them together to achieve an approximation of our original $Z_G(q,\beta)$.

\section{The random cluster model on the discrete hypercube}\label{sec:rccube}

In this section, we prove Theorem~\ref{thm:maincube} about the existence of a structural phase transition for the random cluster model on the discrete hypercube.

\begin{proof}[Proof of Theorem~\ref{thm:maincube}] The proof of part~\ref{part:subcritical} follows from the proof of part~\ref{part:pt1} of Theorem~\ref{thm:pt}, which shows that with high probability, each component in the random cluster model has size at most $(1 + o(1))\ln n \leq \ln_2 n \leq d$.

For part~\ref{part:supercritical}, suppose there are $K$ components with vertex sets $C_1,\dots,C_K$, where $|C_1| \geq \cdots \geq |C_K|$. We make use of the Edwards--Sokal coupling~\cite{ES88}. For a subset of edges distributed according to the random cluster model, color each component at random with one of the $q$ colors to get a configuration of the Potts model. 
 
By~\eqref{eqn:pottsprediction}, every component with more than $\delta n$ vertices should get the same color. If there are at least two such components, then the probability of this occurring is at most $1/q \ll 1$, thus contradicting~\eqref{eqn:pottsprediction}. 

Henceforth, let us assume that there is at most one component of size  at least $\delta n$, that is, $ \max_{i \in \{2,\ldots,K\}}|C_i| \leq \delta n$. Let $n_1 := n - |C_1|$. We have 
\[
\sum_{i=2}^K |C_i|^2 \leq \max_{i \in \{2,\ldots,K\}}|C_i| \cdot \sum_{i = 2}^K|C_i| = \delta n\cdot n_1.
\]
Let $N_j$ be the number of vertices in components $C_2,\ldots C_K$ with color $j$. Clearly, $\mathbb{E}N_j = \frac{n_1}{q}$. Set $t = \sqrt{\delta} \cdot \ln q$. By Hoeffding's inequality, we have that 
\[
\mathbb{P}\left(N_j - \frac{n_1}{q} > t n_1\right) \leq  \exp\left(\frac{-2t^2 n_1^2}{\sum_{i=2}^K |C_i|^2}\right) \leq \exp\left(- \frac{2t^2 n_1}{\delta n}\right).
\]
If $n_1 \leq n/\ln q$, then $|C_1| \geq n\left(1 - 1/\ln q\right)$ and we are done. Else, we have that $\frac{2t^2 n_1}{\delta n} \geq 2\ln q$, and so $\mathbb{P}\left(N_j > (t+1/q) n_1\right) \leq e^{-2\ln q} = 1/q^2$. By taking a union bound over all colors $q$, we have that with high probability all of the $N_j$ are smaller than $(t+1/q) n_1$, and so with high probability, the largest color class has at most 
\[
|C_1| + (t+1/q) n_1 \leq n\left(1 - \Omega\left(\frac{1}{\ln q}\right)\right)
\]
vertices, contradicting~\eqref{eqn:pottsprediction}.
\end{proof}

\section{Conclusion}\label{sec:conclusion}

We conclude with a discussion on future directions and potential extensions of our arguments. 

\subsection{Extension to all temperatures}
A deceptively simple-sounding task would be to extend our algorithms to work at all temperatures, meaning the range $(1-\eps)\beta_0 \leq \beta \leq (1+\eps)\beta_0$. This requires convergence of both our low- and high-temperature cluster expansions beyond the threshold $\beta_0$. 
The main results of~\cite{HJP20} show that this must be possible for large enough $q$ and strong enough expansion conditions, so the main interest here is to close the gap without significantly strengthening our assumptions.

For the high-temperature regime, we can take the same polymer model restricted to polymers on at most $\frac{n}{2}$ vertices and repeat the argument from Section~\ref{sec:hightemp}. Subject to the condition of $\eta$-expansion, the constraint in the polymer size allows for better bounds on the size of $\mathcal P_{u,k,t}$ (the $k$-vertex $t$-edge polymers contained a fixed vertex $u$), and ultimately lets us establish the Koteck{\'y}--Preiss condition for this ``smaller'' polymer model at temperatures slightly beyond $\beta_o$.
There is a delicate interaction between the expansion and improved range of $\beta$, however, and one must handle the colorings that are no longer represented due to the restriction on polymer size.

More significant obstacles appear in extending our low-temperature results. In this regime, our current methods do not extend in a straightforward way beyond $\beta_0$ unless we allow $q$ to be exponential in $d$, thus losing some novelty in our results. 
That is, some extension of what we prove here is plausible, but our conditions seem to degrade and become comparable to the setup of~\cite{HJP20}, while still requiring $q$ to be an integer.
One possible approach to maintaining $q$ polynomial in $d$ is to improve Lemmas~\ref{lem:extremal} and~\ref{lem:TFextremal} by providing a better analysis of the partition functions of $K_{d+1}$ and $K_{d,d}$. 
Even a simple prerequisite such as more a comprehensive version of Lemma~\ref{lem:BGJ} that operates in the entire temperature range, appears not to be known. 
We leave the pursuit of such bounds to future work.

\subsection{Weaker expansion, fewer colors} A natural further direction is to extend our results to a weaker notion of expansion that only guarantees large edge-boundary for small sets, e.g.\ for some $k\ge 2$ we have
\[ \text{Every set}~A~\text{of at most}~n/k~\text{vertices satisfies}~|\nabla(A)| \geq 2|A|. \]
Such a condition is related to the gap between eigenvalues $k$ and $k+1$ of the graph, and it would be interesting to determine whether techniques from combinatorial optimization that handle small-set expansion and such higher-order eigenvalue gaps can be used for approximating partition functions. 
This was initiated in~\cite{CDK20} for the low-temperature polymer model that we study here.

We believe that the lower bounds on $q$ (and to some extent $d$) in our results are artifacts of our techniques. 
Extending our results to the case of all $d,q\ge 3$ presents an interesting challenge, as raised previously in~\cite{HJP20}. The slow mixing results of~\cite{CGG+22} provide some evidence that this may be possible, at least at low temperatures.

\subsection{Conjectured extension to the random cluster model}
It would also be interesting to remove the restriction that $q$ is an integer and handle the random cluster model directly.
That is, are there aspects of our approach that can be extended to work with the more general polymer models of~\cite{HJP20}?

Our high-temperature algorithm already works in the extra generality of the random cluster model, but some of the tools we rely on for the low-temperature algorithm only apply to the Potts model. 
We conjecture that the results we rely on in Section~\ref{sec:extremal} hold in the required generality. 

\begin{conjecture}\label{conj:cliqueextremal}
    For all real $q>1$ and $\beta>0$, for any graph $G$ we have
    \[ Z_G(q,\beta) \le \prod_{v\in V(G)}Z_{K_{d_v+1}}(q,\beta)^{1/(d_v+1)}. \]
\end{conjecture}

\begin{conjecture}\label{conj:bicliqueextremal}
    For all real $q > 1$ and $\beta>0$, for any triangle-free graph $G$ we have
    \[ Z_G(q,\beta) \le \prod_{uv\in V(G)}Z_{K_{d_u,d_v}}(q,\beta)^{1/(d_u d_v)}. \]
\end{conjecture}

Extending our low-temperature algorithm to the random cluster model may not require results as strong as these conjectures, and instead one could work with analogues of Corollaries~\ref{cor:extremal} and~\ref{cor:TFextremal} that offer bounds in terms of the total number of edges and maximum degree parameter. 
So pursuing such slightly weaker results is also of interest.
Note that $q\ge 1$ is required in these conjectures;
they do not hold for $q < 1$ and at $q=1$ we have $Z_G(q,\beta) = e^{\beta e(G)}$ so in this degenerate case the conjectures are trivial.

\section*{Acknowledgments}
This work was undertaken as part of the Phase Transitions and Algorithms working group in the SAMSI Spring 2021 semester program on Combinatorial Probability. We thank the semester organizers for bringing us together. We also thank Will Perkins and the anonymous referees for helpful feedback on this paper.

\bibliographystyle{alpha}
\bibliography{references}
\end{document}